\documentclass[twoside,english]{jfm}
\usepackage[T1]{fontenc}
\usepackage[latin9]{inputenc}
\usepackage[active]{srcltx}
\usepackage{verbatim}
\usepackage{textcomp}
\usepackage{amsmath}
\usepackage{esint}
\usepackage{appendix}
\usepackage{color}

\makeatletter
\numberwithin{equation}{section}
\numberwithin{figure}{section}

\usepackage{graphicx}
\usepackage{natbib}

\ifCUPmtlplainloaded \else
  \checkfont{eurm10}
  \iffontfound
    \IfFileExists{upmath.sty}
      {\typeout{^^JFound AMS Euler Roman fonts on the system,
                   using the 'upmath' package.^^J}%
       \usepackage{upmath}}
      {\typeout{^^JFound AMS Euler Roman fonts on the system, but you
                   dont seem to have the}%
       \typeout{'upmath' package installed. JFM.cls can take advantage
                 of these fonts,^^Jif you use 'upmath' package.^^J}%
      }
  \else
  \fi
\fi

\ifCUPmtlplainloaded \else
  \checkfont{msam10}
  \iffontfound
    \IfFileExists{amssymb.sty}
      {\typeout{^^JFound AMS Symbol fonts on the system, using the
                'amssymb' package.^^J}%
       \usepackage{amssymb}%
         \let\leq=\leqslant
         \let\geq=\geqslant
      }{}
  \fi
\fi

\ifCUPmtlplainloaded \else
  \IfFileExists{amsbsy.sty}
    {\typeout{^^JFound the 'amsbsy' package on the system, using it.^^J}%
     \usepackage{amsbsy}}
    {\providecommand\boldsymbol[1]{\mbox{\boldmath $##1$}}}
\fi

\newcommand{\gr}{\mbox{g}}            

\newsavebox{\astrutbox}
\sbox{\astrutbox}{\rule[-5pt]{0pt}{20pt}}

\newtheorem{theorem}{Theorem}

\newtheorem{corollary}{Corollary}

\title[Fluid particle kinematics of sea waves]{Kinematics of fluid particles on the sea surface. Hamiltonian theory}

\author[F. Fedele, C. Chandre and M. Farazmand]%
{F. Fedele$^{1,2}$%
  \thanks{Email address for correspondence: fedele@gatech.edu},\ns
C. Chandre$^3$ and M. Farazmand$^4$}

\affiliation{$^1$School of Civil and Environmental Engineering, Georgia Institute of Technology,
Atlanta, GA 30322, USA\\[\affilskip]
$^2$School of Electrical and Computer Engineering, Georgia Institute of Technology, Atlanta, GA 30322, USA\\[\affilskip]
$^3$Centre de Physique Th\'eorique - CNRS/Aix-Marseille University, 13009 Marseille, France\\[\affilskip]
$^4$ Center for Nonlinear Sciences, School of Physics, Georgia Institute of Technology,
Atlanta, GA 30332, USA}

\makeatother

\usepackage{babel}
\begin{document}
\maketitle\global\long\def\S{\mathcal{S}}
\global\long\def\eps{\varepsilon}
\global\long\def\H{\mathcal{H}}
\global\long\def\L{\mathcal{L}}
\global\long\def\M{\mathcal{M}}
\global\long\def\K{\mathbf{K}}
\global\long\def\Hilb{\mathbf{H}}
\global\long\def\R{\mathbb{R}}
\global\long\def\ud{\mathrm{d}}

\begin{abstract}
We derive the John-Sclavounos equations, describing the motion of a fluid particle on the sea surface, from first principles using  Lagrangian and Hamiltonian formalisms applied to the motion of a frictionless particle constrained on an unsteady surface. This framework leads to a number of new insights into the particle kinematics.
The main result is that vorticity generated on a stress-free surface vanishes at a wave crest when the horizontal particle velocity equals the crest propagation speed, which is the kinematic criterion for wave breaking. If this holds for the largest crest, then the symplectic two-form associated with the Hamiltonian dynamics reduces instantaneously to that associated with the motion of a particle in free flight, as if the surface did not exist. Further, exploiting the conservation of the Hamiltonian function for steady surfaces and traveling waves, we show that particle velocities remain bounded at all times, ruling out the possibility of the finite-time blowup of solutions. 
\end{abstract}



\section{Introduction}

The horizontal motion of particles of an ideal fluid on a free surface 
obeys a set of nonlinear ordinary differential equations, which only depend on the surface and its space-time gradient and curvature. \cite{John1953} derived the equations of motion for such particles on the zero-stress surface of two-dimensional (2-D) gravity waves, and~\cite{sclav05} generalized them to the three dimensional (3-D) waves. In particular, given a Cartesian
reference system $(x,y,z)$, where $z$ is along the vertical direction,
he exploited the property that the zero-stress free surface $z=\zeta(x,y,t)$
is an iso-pressure surface, and thus the hydrodynamic pressure gradient
$\nabla p$ is collinear with the outward normal $\mathbf{n}\sim\nabla(z-\zeta)$
to the surface, where $\nabla=\left(\partial_{x},\partial_{y},\partial_{z}\right)$. This implies that on the free surface 
\begin{equation}
\nabla(z-\zeta)\times\nabla p=\mathbf{0},\qquad z=\zeta.\label{cross}
\end{equation}
From Euler\textquoteright s equations, the acceleration of a fluid particle in a 3-D flow satisfies 
\begin{equation*}
\frac{{\rm d}^{2}\mathbf{r}}{{\rm d}t^{2}}=-\frac{1}{\rho}\nabla p+\mathbf{f},\label{Euler2}
\end{equation*}
where $\mathbf{r}=(x(t),y(t),z(t))$ is the instantaneous vector
position of the fluid particle and $\mathbf{f}=(0,0,-\gr)$ is the force
due to gravitational acceleration $\gr$. Then, Eq.~(\ref{cross}) can be written as
\begin{equation}
\left(-\partial_x\zeta\mathbf{i}-\partial_y\zeta\mathbf{j}+\mathbf{k}\right)\times\left(-\frac{{\rm d}^{2}\mathbf{r}}{{\rm d}t^{2}}+\mathbf{f}\right)=0,\label{cross2}
\end{equation}
where $({\bf i},{\bf j},{\bf k})$ are unit vectors
along the $x,y$ and $z$ directions, respectively. The $z$ component of the cross
product (\ref{cross2}) is redundant as it is a linear combination
of the $x$ and $y$ components. These yield the coupled equations
\begin{equation}
\begin{array}{c}
\partial_y\zeta\left(\frac{{\rm d}^{2}z}{{\rm d}t^{2}}+\gr\right)+\frac{{\rm d}^{2}y}{{\rm d}t^{2}}=0,\\
\\
\partial_x\zeta\left(\frac{{\rm d}^{2}z}{{\rm d}t^{2}}+\gr\right)+\frac{{\rm d}^{2}x}{{\rm d}t^{2}}=0.
\end{array}\label{JS1}
\end{equation}
Since the fluid particle is constrained on the free surface $\zeta$,
its vertical velocity $\dot{z}=\frac{{\rm d}z}{{\rm d}t}$ and acceleration $\ddot{z}=\frac{{\rm d}^{2}z}{{\rm d}t^{2}}$
depend on the horizontal motion $\mathbf{x}=(x(t),y(t))$. In particular, $\ddot{z}$
follows from differentiating $z(t)=\zeta(x(t),y(t),t)$ with respect
to time. Substituting the resulting $\ddot{z}$ in Eq. (\ref{JS1}) yields the John-Sclavounos (JS) equations~[see Eqs.~(2.17)-(2.18) in~\cite{sclav05}]
\begin{equation}
\begin{array}{c}
\left(1+\zeta_{,x}^{2}\right)\ddot{x}+\zeta_{,x}\zeta_{,y}\ddot{y}+\left(\zeta_{,tt}+\zeta_{,xt}\dot{x}+\zeta_{,yt}\dot{y}+\zeta_{,xx}\dot{x}^{2}+2\zeta_{,xy}\dot{x}\dot{y}+\zeta_{,yy}\dot{y}^{2}+\gr\right)\zeta_{,x}=0,\\
\\
\left(1+\zeta_{,y}^{2}\right)\ddot{y}+\zeta_{,x}\zeta_{,y}\ddot{x}+\left(\zeta_{,tt}+\zeta_{,xt}\dot{x}+\zeta_{,yt}\dot{y}+\zeta_{,xx}\dot{x}^{2}+2\zeta_{,xy}\dot{x}\dot{y}+\zeta_{,yy}\dot{y}^{2}+\gr\right)\zeta_{,y}=0,
\end{array}\label{JS}
\end{equation}
for the evolution of the horizontal fluid particle trajectories driven by the free-surface elevation and its Eulerian temporal and spatial derivatives. \textcolor{black}{Here and in the following, the subscripted commas denote partial derivatives, i.e., $\zeta_{,x}=\partial \zeta/\partial x$.} 
\textcolor{black}{We point out that, as opposed to the Euler's equation, 
the JS equations are a set of ordinary differential equations (ODEs) describing the kinematics of a single fluid particle; as such, they generate a finite-dimensional dynamical system.}

To the best of authors' knowledge, the properties and the structure of the JS equations
have not been investigated in detail. In this work, we derive and study these equations 
using first principles in order to gain mathematical and physical insights into the dynamics
of ocean waves and the inception of wave breaking.

\section{Main findings}
\textcolor{black}{
We demonstrate that the JS equations are more general than initially thought, as they can be derived from first principles 
using Lagrangian and Hamiltonian formalisms. The derivation of~\cite{John1953} assumes that the free surface $\zeta$ 
is generated by an inviscid and irrotational fluid. The derivation of \cite{sclav05}, however, does not assume irrotationality. As we show in section~\ref{sec:vor}, the same equations can be derived from an action principle describing the constrained motion of a frictionless particle on an unsteady surface and subject to gravity. In other words, the 
free surface can be any moving membrane and does not necessarily need to be formed by a fluid.}

Using the Legendre transformation, we also derive the Hamiltonian structure of the JS equations. This 
Hamiltonian structure is also confirmed using Dirac theory as shown in subsection~\ref{sec:Dirac}. The unsteady surface is arbitrary and can originate from many physical processes. In this regard, if we are interested in the fluid particle kinematics on the free surface of gravity water waves, then one needs to know the irrotational flow field that generates a zero-stress free surface separating water from air. Indeed, only if the initial particle velocity is set as that induced by the irrotational flow do the JS equations describe the kinematics of fluid particles. 

\textcolor{black}{Our main result is presented in section~\ref{sec:symplectic}, which required a mathematical description of vorticity created on unsteady free surfaces presented in subsection~\ref{sec:vorticity}.} In particular, we find that vorticity created at a zero-stress free surface vanishes at a wave crest when the horizontal particle velocity equals the propagation speed of the crest. This is the kinematic criterion for wave breaking presented in subsection~\ref{sec:kinematic}~\citep{Perlin2013,Shemer2014,Shemer2015}. Drawing on~\cite{cartan1922lessons}~(chapter II, p. 20), further insights into the particle kinematics are gained by exploiting the relation between the symplectic structure of the JS equations and the physical vorticity as explored by~\cite{Bridges_vorticity_2005} for the shallow water equations. In particular, in subsection~\ref{sec:symvor} our analysis of the Hamiltonian structure of the JS equations reveals that the associated symplectic one-form is the physical fluid circulation and certain terms of the associated two-form relate to the vorticity created on the zero-stress free surface. If the kinematic criterion for wave breaking holds for the largest crest, then the symplectic two-form instantaneously reduces to that associated with the motion of a particle in free flight, as if the free surface and vorticity did not exist.

In this regard, recent studies indicate that the inception of breaking of the largest crest of unsteady wave groups initiates  
when the particle velocity $u_x$ exceeds about $0.84$ times the crest velocity 
$V_c$~\citep{Barthelemy2015,BannerSaket2015}. In particular, none of the non-breaking or recurrent groups reach the 
threshold $B_x=u_x/V_c=0.84$, while all marginal breaking cases exceed the 
threshold~\citep{Barthelemy2015,BannerSaket2015} and eventually the particle speed $u_x$ 
overcomes the wave crest speed $V_c$ (see Figure 3 in~\cite{Barthelemy2015} and~\cite{Shemer2014}). This 
observation motivates a close examination of the space-time transport of wave energy near a large unsteady crest and 
possible local superharmonic instabilities that are triggered as the threshold $B_x$ is exceeded leading to breaking, as 
those found for steep steady waves~\citep{Longuet-HigginspartI1978,BridgesJFMhomoclinic}. 

Our results in section~\ref{sec:slow} suggest that as a wave crest 
grows and approaches breaking, the local kinetic energy $K_{e}$ on the free surface increases much faster than the potential energy $\rho \gr\zeta$ and the normal kinetic energy flux velocity $C_{K_{e}}$ tends to reduce approaching the normal fluid velocity speed $u_{n}$. Equivalently, the Lagrangian kinetic energy flux speed $C_{K_{e}}-u_{n}$ seen by a fluid particle is practically null. Consequently, there is a strong attenuation of accumulation of potential energy on the surface. Thus, at these special instants of time fluid particles on the surface behave like particles in free flight as if the free 
surface did not exist, in agreement with the analysis of the symplectic structure of the particle kinematics. Further studies on the coupling of the kinematics of surface fluid particles with the evolution of the wave field are desirable using Zakharov's (1968) Hamiltonian formalism (\nocite{Zakharov1968}\cite{Krasitskii1994,Zakharov1999}).

Finally, the Hamiltonian formulation of the JS equations also helps gain significant insight into the possibility of singular behavior of particle trajectories and trapping regions, as conjectured by Bridges~(see contributed appendix in~\cite{sclav05}). For instance, in section~\ref{sec:blowup} we exploit the conservation and special form of the Hamiltonian function for steady surfaces and traveling waves and prove that particle velocities stay bounded at all times, ruling out the possibility of the finite-time blowup of solutions. The same argument does not rule out the possible occurrences of finite-time blowups on unsteady surfaces. We also identify regions where particles are  trapped and so remain at all times if their initial velocity is bounded by a prescribed value~(see section~\ref{sec:trapping}).

\section{Hamiltonian properties of the JS equations}\label{sec:vor}
\textcolor{black}{
In the following, we first derive the JS equations from first principles using a Lagrangian formalism applied to the motion of a single frictionless particle constrained on an unsteady surface and subject to gravity (subsection~\ref{sec:lag}). In subsection~\ref{sec:Dirac} we demonstrate that the associated Hamiltonian structure follows from the Legendre transformation and it is also confirmed using Dirac theory of constrained Hamiltonian systems. Finally, in subsection~\ref{sec:symplectic2} the symplectic one- and two-forms are derived. Note that JS equations describe the kinematics of a single inviscid particle; as a result the associated phase-space dynamics is finite-dimensional. 
Further insights into the particle kinematics on a zero-stress free surface are to be gained from the analysis of the symplectic structure of the JS equations and associated differential forms, as discussed in later sections.
}

\subsection{Lagrangian formalism}\label{sec:lag}

The Lagrangian for a free particle subject to gravity in ${\mathbb R}^3$ is given by
\begin{equation*}
\mathcal{L}({\bf r},\dot{\bf r})=K-P,\label{L}
\end{equation*}
where the kinetic and potential energies
\begin{equation*}
K=\frac{1}{2}\left(\dot{x}^{2}+\dot{y}^{2}+\dot{z}^{2}\right),\qquad P=\gr z,\label{KP}
\end{equation*}
and $\mathbf{r}=(x(t),y(t),z(t))$ is the instantaneous vector particle position. Minimizing the action $\mathcal{A}=\int\mathcal{L}{\rm d}t$ over all possible paths yields the Euler--Lagrange equations
\begin{equation*}
\frac{\delta \cal A}{\delta {\bf r}}=\frac{{\rm d}}{{\rm d}t}\left(\frac{\partial \cal L}{\partial\dot{{\bf r}}}\right)-\frac{\partial \cal L}{\partial {\bf r}}=0,\label{Euler}
\end{equation*}
 or equivalently,
$\ddot{\mathbf{r}}=\mathbf{f}$, where $\mathbf{f}=(0,0,-\gr)$. 

We now assume that the particle is constrained  to move on an unsteady surface $z=\zeta(x,y,t)$. Thus, the horizontal particle motion is coupled with that of the evolving surface. The associated dynamical equations follow from the constrained Lagrangian 
\begin{equation}
\mathcal{L}_{\rm c}={\cal L}+\lambda\left[z-\zeta(x,y,t)\right],
\label{Lc}
\end{equation}
where we have introduced the Lagrange multiplier $\lambda$ to
impose that the particle always stays on the surface
$z=\zeta$. Minimizing the action with respect to $x,y,z$ and $\lambda$
yields the set of Euler--Lagrange equations
\begin{align}
\frac{{\rm d}}{{\rm d}t}\frac{\partial {\mathcal L}_{\rm c}}{\partial \dot x}-\frac{\partial {\mathcal L}_{\rm c}}{\partial x} & =\ddot{x}-\lambda_{x}(z-\zeta)+\lambda\zeta_{,x}=0,\label{Leq_1}\\
\frac{{\rm d}}{{\rm d}t}\frac{\partial {\mathcal L}_{\rm c}}{\partial \dot y}-\frac{\partial {\mathcal L}_{\rm c}}{\partial y} & =\ddot{y}-\lambda_{y}(z-\zeta)+\lambda\zeta_{,y}=0,\label{Leq_2}\\
\frac{{\rm d}}{{\rm d}t}\frac{\partial {\mathcal L}_{\rm c}}{\partial \dot z}-\frac{\partial {\mathcal L}_{\rm c}}{\partial z} & =\ddot{z}+\gr-\lambda=0,\label{Leq_3}\\
\frac{\partial {\mathcal L}_{\rm c}}{\partial\lambda} & =z-\zeta=0.  \label{Leq_4}
\end{align}
Here, the last equation imposes the constraint $z=\zeta$, which can be differentiated twice with respect to time to yield the vertical particle velocity 
\begin{equation}
\dot{z}=\zeta_{,x}\dot{x}+\zeta_{,y}\dot{y}+\zeta_{,t},\label{zdot}
\end{equation}
and acceleration
\begin{equation}
\ddot{z}=\zeta_{,x}\ddot{x}+\zeta_{,y}\ddot{y}+\zeta_{,xt}\dot{x}+\zeta_{,yt}\dot{y}+\zeta_{,xx}\dot{x}^{2}+2\zeta_{,xy}\dot{x}\dot{y}+\zeta_{,yy}\dot{y}^{2}+\zeta_{,tt},\label{zdotdot}
\end{equation}
as a function of the horizontal variables $(x,y,\dot{x},\dot{y})$. Then, from Eqs.~(\ref{Leq_1})-(\ref{Leq_2}) the horizontal trajectories satisfy the coupled ordinary differential equations (ODEs)
\begin{equation}
\begin{array}{c}
\ddot{x}+\lambda\zeta_{,x}=0,\\
\\
\ddot{y}+\lambda\zeta_{,y}=0.
\end{array}\label{xysec1}
\end{equation}
The multiplier $\lambda$ satisfies the implicit equation
\begin{equation}
\lambda=\ddot{z}+\gr,\label{Multi}
\end{equation}
which follows from Eq.~\eqref{Leq_3}. In particular, from~Eq.~\eqref{zdotdot},~\eqref{xysec1} the explicit expression for the multiplier follows as
\begin{equation*}
\lambda=\frac{\zeta_{,xt}\dot{x}+\zeta_{,yt}\dot{y}+\zeta_{,xx}\dot{x}^{2}+2\zeta_{,xy}\dot{x}\dot{y}+\zeta_{,yy}\dot{y}^{2}+\zeta_{,tt}+\gr}{1+\zeta_{,x}^2+\zeta_{,y}^2}.\label{lambda}
\end{equation*}
Furthermore, Eqs.~(\ref{xysec1}) can be written as 
\begin{equation*}
\begin{array}{c}
\ddot{x}+\left(\ddot{z}+\gr\right)\zeta_{,x}=0,\\
\\
\ddot{y}+\left(\ddot{z}+\gr\right)\zeta_{,y}=0,
\end{array}\label{xy2}
\end{equation*}
which, after substituting~Eq.~\eqref{zdotdot}, are identical to the JS equations given in~Eq.~\eqref{JS}~(see Introduction). 
%

The JS equations can also be obtained by minimizing the action associated with the reduced Lagrangian 
\begin{equation*}
\mathcal{\widetilde{L}}_c=\frac{1}{2}\left(\dot{x}^{2}+\dot{y}^{2}+\left(\zeta_{,t}+\zeta_{,x}\dot{x}+\zeta_{,y}\dot{y}\right)^{2}\right)-\gr\zeta,\label{Lcons}
\end{equation*}
which follows from the augmented Lagrangian in Eq.~\eqref{Lc} setting $z=\zeta$
and $\dot{z}$ equal to Eq.~\eqref{zdot}. In matrix form
\begin{equation*}
\mathcal{\widetilde{L}}_c=\frac{1}{2}\mathbf{\dot{x}}^{T}\mathbf{B}\dot{\mathbf{x}}+\mathbf{\boldsymbol{\alpha}}^{T}\dot{\mathbf{x}}-\gr\zeta+\frac{1}{2}\zeta_{,t}^{2},\label{Lmat}
\end{equation*}
where $\mathbf{x}=(x(t),y(t))$ is the horizontal vector of position and 
\begin{equation}
\mathbf{B}=\left[\begin{array}{cc}
1+\zeta_{,x}^{2} & \zeta_{,x}\zeta_{,y}\\
\zeta_{,x}\zeta_{,y} & 1+\zeta_{,y}^{2}
\end{array}\right],\qquad\boldsymbol{\alpha}=\zeta_{,t}\left[\begin{array}{c}
\zeta_{,x}\\
\zeta_{,y}
\end{array}\right].\label{Bmat}
\end{equation}
We note that $\mathbf{B}$ is symmetric and positive-definite with real eigenvalues  
\begin{equation*}
\lambda_1=1,\quad\quad \lambda_2=|\mathbf{B}|=1+\zeta_{,x}^2+\zeta_{,y}^2,
\label{eig}
\end{equation*}
and the corresponding orthogonal eigenvectors
\begin{equation*}
\mathbf{w}_1=(-\zeta_{,y},\zeta_{,x})=\nabla^{\perp}\zeta,\quad\quad \mathbf{w}_2=(\zeta_{,x},\zeta_{,y})=\nabla\zeta.
\label{eigv}
\end{equation*}
These will be useful later in the analysis of the finite time blowup of the JS equations (cf. Section~\ref{sec:blowup}). 

The generalized momentum $\mathbf{p}=(p_x,p_y)$ is a function of the horizontal particle velocity $\mathbf{\dot{x}}$ as
\begin{equation}
\mathbf{p}=\mathbf{B}\mathbf{\dot{x}}+\boldsymbol{\alpha},
\label{pmon1}
\end{equation}
where
\begin{equation}
p_{x}=\frac{\partial\mathcal{\widetilde{L}}_c}{\partial\dot{x}}=\left(1+\zeta_{,x}^{2}\right)\dot{x}+\zeta_{,x}\zeta_{,y}\dot{y}+\zeta_{,x}\zeta_{,t},
\label{px1}
\end{equation}
and
\begin{equation}
p_{y}=\frac{\partial\mathcal{\widetilde{L}}_c}{\partial\dot{y}}=\left(1+\zeta_{,y}^{2}\right)\dot{y}+\zeta_{,y}\zeta_{,x}\dot{x}+\zeta_{,y}\zeta_{,t}.
\label{py1}
\end{equation}
Then $(\mathbf{p},\mathbf{x})$ are canonically conjugate variables and the Hamiltonian follows from the Legendre transform of $\mathcal{\widetilde{L}}$ as~\citep{Morrison1998}
\begin{equation}
{\mathcal{H}}_c=p_{x}\dot{x}+p_{y}\dot{y}-\mathcal{\widetilde{L}}=\mathbf{p}^T\mathbf{\dot{x}}-\mathcal{\widetilde{L}}.
\label{HC}
\end{equation}
From Eq.~(\ref{pmon1}) the horizontal particle velocity $\mathbf{\dot{x}}$ can be written as a function of the canonical momentum $\mathbf{p}$, and the Hamiltonian can be recast as 
\begin{equation}
\mathcal{H}_c=\frac{1}{2}\left(\mathbf{p}-\boldsymbol{\alpha}\right)^{T}\mathbf{B}^{-1}\left(\mathbf{p}-\boldsymbol{\alpha}\right)+\gr\zeta-\frac{1}{2}\zeta_{,t}^{2}.\label{HV}
\end{equation}
Note that for unsteady surfaces, $\mathcal{H}_c$ is not conserved as particles behave as an open system exchanging energy with the moving surface.

The Lagrangian formalism developed above highlights a fundamental property of the JS equations.  
On the one hand, these are originally derived from the dynamical condition that the zero-stress free surface $z=\zeta$ is an iso-pressure surface~\citep{sclav05}. On the other hand, we have derived the same equations from an action principle for the constrained motion of a frictionless particle subject to gravity on an unsteady surface. The unsteady surface is arbitrary and can be generated by many physical processes. If the interest is in the kinematics of fluid particles on the free surface of gravity water waves, one must know the irrotational velocity field beneath the waves. Indeed, only if the initial particle velocity is set as that induced by the irrotational flow do the JS equations describe the kinematics of fluid particles. 

A rigorous proof of the previous statement is beyond the scope of this paper. We only point out that the horizontal velocity $\dot{\mathbf{x}}$ of a fluid particle on an irrotational water surface satisfies
\begin{equation}
\dot{\mathbf{x}}=\mathbf{U}_h(\mathbf{x}(t),\zeta(x,y,t),t),\label{xdot}
\end{equation}
where the horizontal Eulerian velocity $\mathbf{U}_h=\nabla\phi=(\phi_{,x},\phi_{,y})$ is given in terms of the velocity potential $\phi(x,y,z,t)$. 
Thus, we expect that the JS equations~\eqref{JS} can also be derived using Eq.~\eqref{xdot} and the Stokes equations (see section \ref{sec:kinematic}, and in particular Eqs.~\eqref{B},~\eqref{B1a}). For instance, the JS equations for the case of steady irrotational flows are derived in Appendix~\ref{app:Stokes}. 

\subsection{Hamiltonian formalism via Dirac Theory}\label{sec:Dirac}

The Lagrangian formalism developed in the previous section yields the Hamiltonian structure of the JS equations~(\ref{JS}) in terms of the canonical variables $(\mathbf{p},\mathbf{x})$. A non-canonical structure in terms of the original physical variables (position $\mathbf{x}$ and velocity $\mathbf{u}$) can be derived within  the framework of Dirac's (1950) theory of constrained Hamiltonian systems (see also~\cite{dira58}). The transformation~\eqref{pmon1} between the non-canonical and canonical variables follows from Darboux's theorem for finite-dimensional Hamiltonian systems~(see, e.g.,~\cite{Morrison1998}).

\subsubsection{Dirac theory: an introduction}

An alternative way to constrain a Hamiltonian system is to work directly within the Hamiltonian structure and consider Lagrange multipliers associated with the constraints on the Hamiltonian
$$
H_*=H+\lambda_\alpha \Phi_\alpha,
$$
where $\lambda_\alpha$ are the Lagrange multipliers, $\Phi_\alpha$ are the constraints and with an implicit summation over $\alpha$ which labels the constraints. In the case under consideration, there are two constraints: the first one is to impose that the particle is on the surface at a given time (i.e., $z=\zeta$), and the second one is to impose that the velocity of the particle coincides with the velocity of the surface at the given position and the given time (i.e., $u_z=\dot{z}$ equals $\mathrm{d}\zeta/\mathrm{d}t$). The advantage of working within the Hamiltonian framework is to obtain the expression of the constrained system within the same set of dynamical variables. For instance, in the case we consider the dynamical variables are the positions and the velocities of the particles. Imposing the constraints within the Hamiltonian framework allows one to obtain the constrained dynamics also in terms of positions and velocities. 
In a very similar way as the Lagrangian framework, the Lagrange multipliers are obtained by imposing that the constraints are conserved quantities in the dynamics given by $H_*$, i.e., ${\mathrm d}{\Phi_\alpha}/{\mathrm d}t=0$.  

Consider a parent (unconstrained) Hamiltonian system defined by the Poisson bracket 
\begin{equation}
\{F,G\}=\nabla F \cdot {\mathbb J}({\bf z})\nabla G,
\label{bracket}
\end{equation}
and Hamiltonian $\mathcal{H}({\bf z})$ with dynamical variables ${\bf z}=(z_1,\ldots,z_N)$, where ${\mathbb J}({\bf z})$ is the $N\times N$ Poisson matrix and  $\nabla =(\partial_{z_1},\ldots,\partial_{z_N})$.  
We recall that the Poisson bracket is an antisymmetric bilinear operator
\begin{equation}
\{F,G\}=-\{G,F\},\label{bilin}
\end{equation}
it satisfies the Leibniz rule 
\begin{equation}
\{F_1F_2,F_3\}=F_1\{F_2,F_3\}+\{F_1,F_3\}F_2,\label{Leib}
\end{equation} 
and the Jacobi identity
\begin{equation}
\{\{F_1,F_2\},F_3\}+\{\{F_3,F_1\},F_2\}+\{\{F_2,F_3\},F_1\}=0,\label{Jacob}
\end{equation}
for all observables $F_1({\bf z})$, $F_2({\bf z})$ and $F_3({\bf z})$ scalar functions of the dynamical variables. 

For the particle kinematics on a free surface, the dynamical variables are ${\bf z}=(x,y,z,u_x,u_y,u_z)$ and the Poisson matrix is the canonical one:
$$
{\mathbb J}=\left( \begin{array}{cc} 0 & {\mathbb I}_3\\ -{\mathbb I}_3 & 0  \end{array} \right),
$$
leading to the well-known Hamilton's equation from the equations of motion of any observable $F$ given by ${\mathrm d}F/{\mathrm d}t=\{F,H\}$ for the unconstrained dynamics generated by $H$ or by ${\mathrm d}F/{\mathrm d}t=\{F,H_*\}$ for the constrained dynamics generated by $H_*$.

The Lagrange multipliers are obtained from $\{\Phi_\alpha,H_*\}=0$ and are defined by the set of equations 
$$
\{\Phi_\alpha,\Phi_\beta\}\lambda_\beta+\{\Phi_\alpha, H\}=0,
$$ 
using the bilinearity of the Poisson bracket in~Eq.~\eqref{bilin} and the associated Leibniz rule in~Eq.~\eqref{Leib}. This equation is valid on the surface defined by the constraints $\Phi_\alpha=0$. In order to solve for the Lagrange multipliers, we define the matrix ${\mathbb C}$ with elements $C_{\alpha \beta}=\{\Phi_\alpha , \Phi_\beta\}$. If this matrix is invertible, we denote ${\mathbb D}$ with elements $D_{\alpha \beta}$ its inverse, and the Lagrange multipliers are given by $\lambda_\beta =-D_{\beta \gamma}\{\Phi_\gamma,H\}$. Therefore the equations of motion ${\mathrm d}F/{\mathrm d}t=\{F,H_*\}$ in the constrained system become
\begin{equation}
\label{eq:FHD}
\dot{F}=\{F,H\}-\{F,\Phi_\alpha\}D_{\alpha \beta}\{\Phi_\beta,H\},
\end{equation}
using again the bilinearity and the Leibniz rule of the Poisson bracket $\{\cdot,\cdot\}$~(see Eqs.~\eqref{bilin} and \eqref{Leib}). In the same way as above, these equations of motion are valid on the surface defined by the constraints $\Phi_\alpha=0$. 

Following~\cite{dira50,dira58}, Eq.~(\ref{eq:FHD}) suggests to define a new bracket for the constrained system as
\begin{equation}
\label{eq:expDB}
\{F,G\}_*=\{F,G\}-\{F,\Phi_{\alpha}\}D_{\alpha \beta}\{\Phi_{\beta},G\},
\end{equation}
such that the equations of motion for the constrained system are given by ${\mathrm d}F/{\mathrm d}t=\{F,H\}_*$, i.e., with the original Hamiltonian $H$ but a different bracket. The highly non-trivial feature is that this bracket is a Poisson bracket, i.e., it satisfies the Jacobi identity, as it was proved by Dirac. As a consequence, the constrained system defined by the Hamiltonian $H$ and the bracket $\{\cdot,\cdot\}_*$ is a Hamiltonian system. 

\subsubsection{Non-canonical Hamiltonian of the JS equations}

The two constraints we consider are explicitly written as
\begin{equation}
\Phi_1=z-\zeta(x,y,t)=0,\quad\quad \Phi_2=u_z-u_x\zeta_{,x}-u_y \zeta_{,y}-\zeta_{,t}=0.
\label{con}
\end{equation}
The matrix ${\mathbb C}$ is invertible since 
\begin{equation}
\label{eq4C}
C_{11}=C_{22}=0,\quad\quad C_{12}=-C_{21}=\{\Phi_1,\Phi_2\}=1+\zeta_{,x}^2+\zeta_{,y}^2.
\end{equation}
The Dirac bracket~\eqref{eq:expDB} specializes to
\begin{equation} 
\{F,G\}_*=\nabla F \cdot \overline{\mathbb J}_*\nabla G,\label{dbracket} 
\end{equation}
where $\nabla=\partial/\partial {\bf z}$ and ${\bf z}=(x,y,t,u_x,u_y,E)$. The Poisson matrix is given by 
\begin{equation}
\label{eq:DBzeta}
\overline{\mathbb J}_*=\left( \begin{array}{cc} 0 & \overline{\bf B}^{-1}\\ -(\overline{\bf B}^{-1})^T & \overline{\cal B}  \end{array} \right), 
\end{equation}
with 
$$
\overline{\bf B}= \left( \begin{array}{ccc} 1+\zeta_{,x}^2 & \zeta_{,x} \zeta_{,y} & \zeta_{,x}\zeta_{,t} \\ \zeta_{,x}\zeta_{,y}  & 1+\zeta_{,y}^2  & \zeta_{,y}\zeta_{,t} \\ 
0 & 0  &  1 \end{array} \right),
$$
and
$$
\overline{\cal B}=\left( \begin{array}{ccc} 0 & -b_3 & b_2 \\ b_3  & 0  & -b_1 \\ 
-b_2 & b_1 &  0 \end{array} \right).
$$
The vector ${\bf b}_{\rm m}=(b_1,b_2,b_3)$ given by
\begin{equation}
{\bf b}_{\rm m}=\frac{\overline{\nabla} \zeta \times \overline{\nabla}\left( u_x\zeta_{,x}+u_y\zeta_{,y}+\zeta_{,t} \right)}{1+\vert \nabla\zeta \vert^2}
=\frac{\overline{\nabla} \zeta \times\left[ \left(\overline u\cdot\overline{\nabla}\right) \overline{\nabla}\zeta\right]}{1+\vert \nabla\zeta \vert^2}.\label{bb}
\end{equation}
Here $\overline{\nabla}$ designates the gradient in space-time variables $(x,y,t)$ whereas $\nabla$ is the gradient in space variables $(x,y)$ and $\overline u=(u_x,u_y,1)$. The matrix $\overline{\bf B}$ is always invertible and its eigenvalues are $1+\zeta_{,x}^2+\zeta_{,y}^2$ and $1$ (of multiplicity two). The dynamical variable $E$ is canonically conjugate to time and corresponds to an energy variable, the amount of energy brought in by the time-dependence of the surface. 
More details on the computation of the Dirac bracket is given in Appendix~\ref{app:Dirac}. 

The Hamiltonian formulation of the reduced bracket in the physical variables $(x,y,t,u_x,u_y,E)$ is non-canonical. The constrained Hamiltonian $\overline{\mathcal H}_c$ is obtained from the free-particle Hamiltonian , replacing $z$ by $\zeta$ and $u_z$  by $u_x\zeta_{,x}+u_y\zeta_{,y}+\zeta_{,t}$ (see Appendix~\ref{app:Dirac})
\begin{equation}
\label{Hc}
\overline{\mathcal H}_c=\frac{u_x^2+u_y^2+(\zeta_{,x} u_x+\zeta_{,y} u_y+\zeta_{,t})^2}{2}+\gr \zeta +E.
\end{equation}
Then, the equations of motion are given by
\begin{equation}
\frac{{\rm d} \overline{F}}{{\rm d} \tau}=\{\overline{F},\overline{\mathcal{H}}_c\}_*,
\label{Fbrack}
\end{equation}
where $\overline{F}$ is any function of the dynamical variables. It follows that, as expected,
$$
\frac{{\rm d}t}{\rm d \tau}=\{t,\overline{\mathcal{H}}_c\}_*=1,
$$ 
i.e., $t=\tau$ with a proper choice of the initial time. Then, the JS equations~\eqref{JS} are given by $\dot{x}=u_x$ and $\dot{y}=u_y$ and
$$
\frac{{\rm d}u_x}{{\rm d} t}=\{u_x,\overline{\mathcal{H}}_c\}_*,\quad\quad \frac{{\rm d}u_y}{{\rm d} t}=\{u_y,\overline{\mathcal{H}}_c\}_*.
$$
Furthermore, we get an equation for the evolution of the energy $E$ as
$$
\frac{\dot{E}}{\zeta_{,t}}=\frac{\dot{u}_x}{\zeta_{,x}}=\frac{\dot{u}_y}{\zeta_{,y}}.\label{zz}
$$

For a time-independent surface, the Poisson bracket can be further simplified, since the variables $(t,E)$ can be dropped. In this case, the Poisson matrix reduces to a $4\times 4$ matrix
$$
{\mathbb J}_1=\left( \begin{array}{cc} 0 & {\bf B}^{-1}\\ -({\bf B}^{-1})^\dagger & {\cal B}  \end{array} \right), 
$$
where ${\bf B}$ is given by Eq.~(\ref{Bmat}) and
$$
{\cal B}=b_3\left( \begin{array}{cc} 0 & -1\\ 1 & 0 \end{array} \right).
$$

\subsubsection{Canonical Hamiltonian via Darboux theorem }
\textcolor{black}{
Following Darboux's theorem for finite-dimensional Hamiltonian systems~(see, e.g.,~\cite{Morrison1998}), it is possible to transform the Poisson bracket defined by the 
Poisson matrix~(\ref{eq:DBzeta}) into a canonical form. In principle the canonical and non-canonical coordinates 
are equivalent. In practice, however, utilizing one is favored over the other.
For instance, working with physical variables has the advantage of lending itself to a better intuition.
Working with a canonical bracket, on the other hand, has its own advantages, e.g., allowing the use of symplectic 
algorithms developed for finite-dimensional canonical Hamiltonian systems. 
}

Here we apply Darboux's algorithm by modifying the momenta $u_x$, $u_y$ and $E$. In order to find the new momenta $p_x$, $p_y$ and $\tilde{E}$ which are canonically conjugate to $x$, $y$ and $t$ respectively, one has to solve first order linear partial differential equations of the kind $\{x,p_x\}=1$, e.g., using the method of characteristics. We restrict the search of these new variables to $p_x=p_x(x,y,t,u_x,u_y)$, $p_y=p_y(x,y,t,u_x,u_y)$ and $\tilde{E}=E+\varepsilon(x,y,t,u_x,u_y)$. 
After some algebra, the change of variables reads
\begin{eqnarray}
&& p_x=(1+\zeta_{,x}^2)u_x+\zeta_{,x}\zeta_{,y} u_y +\zeta_{,x}\zeta_{,t},
\nonumber \\
&& p_y=\zeta_{,x}\zeta_{,y} u_x+(1+\zeta_{,y}^2) u_y+\zeta_{,y}\zeta_{,t},\label{darboux}\\
&& \tilde{E}=E+\zeta_{,t}(u_x\zeta_{,x}+ u_y\zeta_{,y}+\zeta_{,t}).
\nonumber 
\end{eqnarray}
The first two equations yield the generalized momentum $\mathbf{p}=(p_x,p_y)$ as a function of the horizontal particle velocity $\mathbf{u}_h=\left(u_x,u_y\right)$ as in Eq.~\eqref{pmon1}, i.e.\ $\mathbf{p}=\mathbf{B}\mathbf{u}_h+\boldsymbol{\alpha}$, where $ \boldsymbol{\mathbf \alpha}$ and ${\bf B}$ are given by Eq.~(\ref{Bmat}). The Hamiltonian~\eqref{Hc} in terms of the canonically conjugate variables $\left(\mathbf{x},t\right)$ and $(\mathbf{p},\tilde{E})$ becomes
\begin{equation*}
\overline{\mathcal{H}}_c=\frac{1}{2}({\bf p}-\boldsymbol{\alpha}) \cdot {\bf B}^{-1}({\bf p}-\boldsymbol{\alpha})+\gr\zeta -\frac{\zeta_{,t}^2}{2}+\tilde{E}.
\label{Hdarboux}
\end{equation*}
This coincides with the Hamiltonian in Eq.~\eqref{HC} derived from the Lagrangian formalism, except for the extra variable $\tilde{E}$, canonically conjugate of the time $t$. The former is needed to make the system autonomous, as $\tilde{E}$ is the energy that the particle exchanges with the moving surface. 

Concerning the one-dimensional case, e.g., when $\zeta_{,y}=0$, the Hamiltonian simplifies to
$$
\mathcal{H}_c=\frac{({p_x}-\zeta_{,t}\zeta_{,x})^2}{2(1+\zeta_{,x}^2)}+\frac{{p_y}^2}{2}+\gr\zeta-\frac{\zeta_{,t}^2}{2}+\tilde{E}.
$$
Since the potential does not depend on $y$, the momentum $p_y$ is constant, so the motion in the $y$-direction is trivial. In the non-trivial direction, the reduced one-dimensional Hamiltonian becomes
$$
\mathcal{H}_{1D}=\frac{(p_x-\zeta_{,t}\zeta_{,x})^2}{2(1+\zeta_{,x}^2)}+\gr\zeta-\frac{\zeta_{,t}^2}{2}, 
$$
where we have removed $\tilde{E}$ to consider the non-autonomous Hamiltonian (which is now not conserved). 

In the time-independent case ($\zeta_{,t}=0$), the additional variables $(t,E)$ can be eliminated since the set of observables $F(x,y,{p_x},{p_y})$ constitutes a Poisson sub-algebra. The resulting Hamiltonian then reads
$$
\mathcal{H}_c=\frac{1}{2}{\bf p}\cdot {\bf B}^{-1}{\bf p}+\gr\zeta,
$$
and ${\bf p}={\bf B} {\bf u}_h$. This Hamiltonian resembles the one of the free particle, except that the metric for the kinetic energy is defined by ${\bf B}^{-1}$. 

Another case of interest is the traveling wave $\zeta(x,y,t)=\overline{\zeta}(x-ct,y)$. Changing the dynamics to the moving frame with velocity $c$ is a time-dependent change of coordinates, so it has to be performed in the autonomous framework. We perform a canonical transformation defined by $\overline{x}=x-ct$ and $\overline{E}=E+c {p_x}$, the other variables remain unchanged. Being canonical, this change of variables does not modify the expression of the bracket. The reduced (time-independent) Hamiltonian becomes
\begin{equation}
\mathcal{H}_c=\frac{1}{2}({\bf p}-\boldsymbol{\alpha}) \cdot {\bf B}^{-1}({\bf p}-\boldsymbol{\alpha})+\gr\overline{\zeta} -c^2\frac{\overline{\zeta}_x^2}{2}-c{p_x},
\label{eq:H_tw}
\end{equation}
with $\boldsymbol{\alpha}=-c\overline{\zeta}_x(\overline{\zeta}_x,\overline{\zeta}_y)$, and the canonically conjugate variables are $(\overline{x},{p_x})$ and $(y,{p_y})$. Here, the matrix ${\bf B}$ is given by Eq.~(\ref{Bmat}) where $\zeta$ is replaced by $\overline{\zeta}$. 

Hamiltonian~\eqref{eq:H_tw} can be written in the form
\begin{equation}
\mathcal{H}_c=\frac{1}{2}({\bf p}-\boldsymbol{\alpha}-c\mathbf B\mathbf e_1) \cdot {\bf 
B}^{-1}({\bf p}-\boldsymbol{\alpha}-c\mathbf B\mathbf e_1)+\gr\overline{\zeta} 
-\frac{1}{2}c^2,
\label{eq:H_tw2}
\end{equation}
with $\mathbf e_1=(1,0)^T$. Next, we express the Hamiltonian in terms of the 
particle velocity in the co-moving frame, $\overline{\mathbf u}_h=(\dot{\overline x},\dot y)$.
From the fact that $\dot{\overline x}=\partial \mathcal{H}_c/\partial p_x$ and
$\dot{y}=\partial \mathcal{H}_c/\partial p_y$, we have
\begin{equation*}
\overline{\mathbf u}_h=\mathbf B^{-1}\left({\mathbf p}-\pmb{\alpha}-c\mathbf B\mathbf e_1
\right).
\end{equation*}
Substitution in Eq.~\eqref{eq:H_tw2} yields,
\begin{equation}
\mathcal{H}_c=\frac{1}{2}\overline{\bf u}_h\cdot
{\bf B}\overline{\bf u}_h+\gr\overline{\zeta} -\frac{1}{2}c^2.
\label{Htw3}
\end{equation}
This form of the Hamiltonian will prove helpful in our analysis of the finite time blowup of the JS equations.

{\em Remark: Physical interpretation of the vector ${\bf b}_{\rm m}$ in Eq.~\eqref{bb}.}
The Poisson structure of particle motion on an unsteady surface bears some similarities with the motion of a charged particle in electromagnetic fields. In terms of the physical variables (position $\bf x$ and velocity $\bf u$), the Poisson bracket of a charge particle in a magnetic field is non-canonical with a part of the form ${\bf b}_{\rm m}\cdot (\partial_{\bf u}F \times \partial_{\bf u} G)$, called gyrobracket (responsible for the gyration motion of the particle around magnetic field lines) where ${\bf b}_{\rm m}$ is the magnetic field. In canonical coordinates, the velocity ${\bf u}$ has to be shifted by the vector potential ${\bf A}_{\rm m}$, which satisfies ${\bf b}_{\rm m}=\overline{\nabla} \times {\bf A}_{\rm m}$ (see~\cite{litt79} for more details). 

Our vector ${\bf b}_{\rm m}$ in Eq.~\eqref{bb} can be interpreted as a magnetic field in the extended phase space and the associated vector potential follows from
\begin{equation*}
(1+\vert \nabla\zeta \vert^2) {\bf b}_{\rm m}=\overline{\nabla} \times {\bf A}_{\rm m},
\end{equation*}
with
\begin{equation*}
{\bf A}_{\rm m}=-(u_x\zeta_{,x}+u_y\zeta_{,y}+\zeta_{,t})\overline{\nabla} \zeta.
\end{equation*} 
Notice that in general ${\bf b}_{\rm m}$ is not divergence-free because of the factor $(1+\vert \nabla\zeta \vert^2)$. Furthermore, in Eq.~(\ref{eq:DBzeta}), the term ${\cal B}$ generates a term ${\bf b}_{\rm m}\cdot \partial_{\bf u}F \times \partial_{\bf u} G$ in the Poisson bracket, since we notice that $\overline{\cal B}$ can be written as $\overline{\cal B}={\bf b}_{\rm m}\times$, i.e.\ it maps a vector ${\bf v}$ into ${\bf b}_{\rm m}\times {\bf v}$. Whereas when the Poisson bracket is canonical, the momenta have to be shifted by the ''vector potential'' $\boldsymbol{\alpha}$ [see Eq.~\eqref{pmon1}].

\subsection{Symplectic structure}\label{sec:symplectic2}
 
The symplectic one-form
\begin{equation}
\omega^1={p_x}{\rm d}x+{p_y}{\rm d}y+\tilde{E}{\rm d}t
\label{C_1_form}
\end{equation}
is given in terms of the canonically conjugate variables $(\mathbf{x},t,{\mathbf{p}},\tilde{E})$. The associated two-form $\omega^2={\rm d}\omega^1$, which provides the symplectic structure of the dynamics, follows by exterior differentiation of~Eq.~\eqref{C_1_form} as
\begin{equation}
\omega^2={\rm d}{p_x}\wedge {\rm d}x+{\rm d} {p_y}\wedge {\rm d}y+{\rm d}\tilde{E}\wedge {\rm d}t.\label{2form}
\end{equation}
To gain physical insights into the inviscid kinematics of fluid particles near large crests, it is convenient to write the above symplectic forms in terms of the non-canonical variables $\overline{\bf z}=(x,y,u_x,u_y,t,E)$.
Using the transformations~\eqref{darboux}, Eq.~\eqref{C_1_form} yields
\begin{eqnarray}
\omega^{1} & = & \left((1+\zeta_{,x}^{2})u_{x}+\zeta_{,x}\zeta_{y}u_{y}+\zeta_{,x}\zeta_{,t}\right){\rm d}x+\left(\zeta_{x}\zeta_{,y}u_{x}+(1+\zeta_{,y}^{2})u_{y}+\zeta_{,y}\zeta_{,t}\right){\rm d}y\nonumber\\
 &  & \qquad\qquad+\left(E+\zeta_{,t}(\zeta_{,x}u_{x}+\zeta_{,y}u_{y}+\zeta_{,t})\right){\rm d}t,\label{omega1} 
\end{eqnarray}
and Eq.~\eqref{2form} becomes
\begin{eqnarray}
\omega^{2} & = & -(1+\vert\nabla\zeta\vert^{2})b_{3}{\rm d}x\wedge{\rm d}y+
(1+\vert\nabla\zeta\vert^{2})b_{2}{\rm d}x\wedge{\rm d}t-(1+\vert\nabla\zeta\vert^{2})b_{1}{\rm d}y\wedge{\rm d}t\nonumber\\
& &+(1+\zeta_{,x}^{2}){\rm d}u_{x}\wedge{\rm d}y+\zeta_{,x}\zeta_{,y}{\rm d}u_{y}\wedge{\rm d}x+\zeta_{,x}\zeta_{,y}{\rm d}u_{x}\wedge{\rm d}y+(1+\zeta_{,y}^{2}){\rm d}u_{y}\wedge{\rm d}y\nonumber\\
& & +\zeta_{,x}\zeta_{,t}{\rm d}u_{x}\wedge{\rm d}t+\zeta_{,y}\zeta_{,t}{\rm d}u_{y}\wedge{\rm d}t+{\rm d}E\wedge{\rm d}t.\label{omega2}
\end{eqnarray}
Note that the two-form can also be obtained from the Lagrange matrix as
\begin{equation*}
\omega^2=\overline{L}_*^{\alpha \beta} {\rm d}z_\alpha \wedge {\rm d}z_\beta/2,
\label{w2q}
\end{equation*}
where $\overline{L}_*^{\alpha \beta}$ is the inverse of the Dirac-Poisson matrix~\eqref{eq:DBzeta}, that is
$$
\overline{\mathbb L}_*=\left( \begin{array}{cc} (1+\vert \nabla \zeta\vert^2)\overline{\cal B} & -\overline{\bf B}^T\\ \overline{\bf B} & 0  \end{array} \right).
$$

\section{Physical interpretation of the symplectic structure}\label{sec:symplectic}
\textcolor{black}{ 
In this section, we study in detail the symplectic structure of the JS equations obtained above. In particular, we provide a physical interpretation of the one- and two-forms~\eqref{omega1}~and~\eqref{omega2} in terms of circulation and vorticity created on the zero-stress free surface~(\cite{cartan1922lessons}, chapter II, p. 20, see also~\cite{Bridges_vorticity_2005}).
}

\textcolor{black}{
First, in subsection~\ref{sec:vorticity} we present the mathematical description of vorticity generated on a zero-stress free surface. In particular, we draw on~\cite{Longuet_Higgins_curvature1998} and extend his formulation for steady surfaces to the unsteady case. The associated velocity circulation is also derived. Then in subsection~\ref{sec:kinematic} we show that the classical kinematic criterion for wave breaking~\citep{Perlin2013} follows from the condition of vanishing vorticity at a wave crest.
}

\textcolor{black}{
Finally, in subsection~\ref{sec:symvor} our analysis reveals that the symplectic one-form of the JS equations obtained in section \ref{sec:symplectic2} is the physical fluid circulation and certain terms of the associated two-form relate to the vorticity created on the zero-stress free surface. Furthermore, if the kinematic criterion for wave breaking holds for the largest crest, then the two-form instantaneously reduces to that associated with the motion of a particle in free flight, as if the free surface and vorticity did not exist.
}

\subsection{Vorticity generated at a zero-stress free surface}\label{sec:vorticity}

\begin{figure}
\centering
\includegraphics[width=0.8\textwidth]{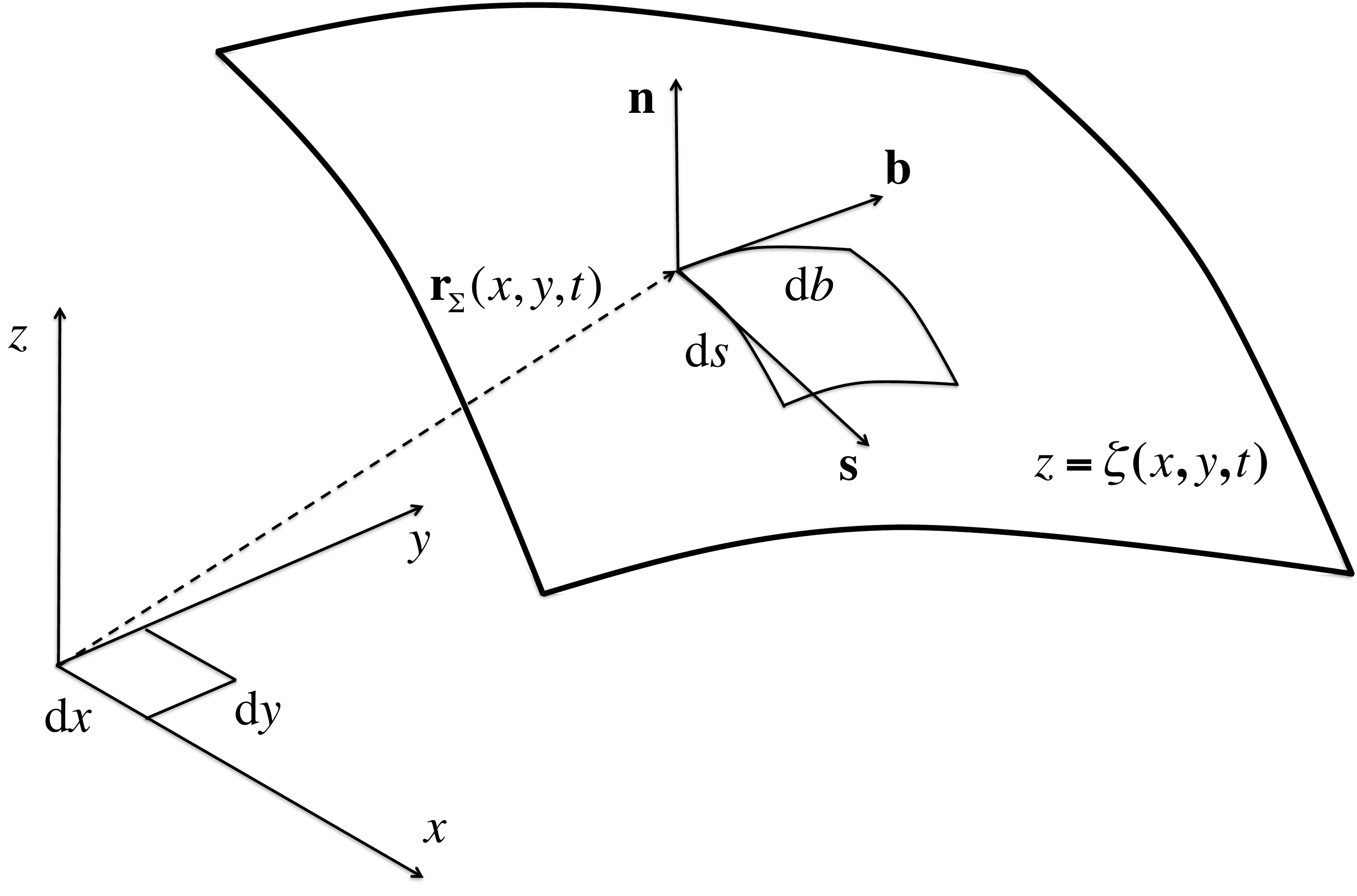}
\caption{Reference coordinate system: in the global frame $(x,y,z)$, 
$\boldsymbol{\mathbf{r}}_{\Sigma}(x,y,t)$ is a point of the free surface $z=\zeta(x,y,t)$, and 
$(\mathbf{s},\mathrm{\mathbf{b}},\mathbf{n})$ is a local frame on the surface. 
}
\label{FIG1} 
\end{figure}


In general, vorticity is generated at free surfaces whenever there
is flow past regions of surface curvature~\citep{Wu1995,Lundgren}.
This non-zero vorticity resides in a vortex sheet along the
free-surface even when the flow field beneath the
free surface is irrotational~\citep{Longuet_Higgins_curvature1998}.
The condition of zero shear stress determines the strength of the
vorticity at the surface. In the global frame $(x,y,z)$, a point
$\boldsymbol{\mathbf{r}}_{\Sigma}$ of the free-surface $\Sigma$ can be parametrized
as 
\[
\boldsymbol{\mathbf{r}}_{\Sigma}(x,y,t)=\left(\begin{array}{c}
x\\
y\\
\zeta(x,y,t)
\end{array}\right),
\]
where $x$ and $y$ are the parameters. Here we consider single valued surfaces so that $z=\zeta(x,y,t)$ is well defined at any point $(x,y)$ and time $t$.
The local frame $(\mathbf{s},\mathrm{\mathbf{b}},\mathbf{n})$ on the surface
is given by
\[
\mathbf{s}=\frac{\partial_{x}\mathbf{r}_{\Sigma}}{\left|\partial_{x}\mathbf{r}_{\Sigma}\right|},\qquad\mathbf{b}=\frac{\partial_{y}\mathbf{r}_{\Sigma}}{\left|\partial_{y}\mathbf{r}_{\Sigma}\right|},\qquad\mathbf{n}=\frac{\partial_{x}\mathbf{r}_{\Sigma}\times\partial_{y}\mathbf{r}_{\Sigma}}{\left|\partial_{x}\mathbf{r}_{\Sigma}\times\partial_{y}\mathbf{r}_{\Sigma}\right|},
\]
where $\boldsymbol{\mathbf{s}}$ and $\boldsymbol{\mathbf{b}}$ are
unit vectors tangent to the surface and $\mathbf{n}$ is the unit
vector of the outward normal (see Fig.~\ref{FIG1}). More explicitly,
\begin{equation}
\boldsymbol{\mathbf{s}}=\frac{1}{\sqrt{h_1}}\left(\begin{array}{c}
1\\
0\\
\zeta_{,x}
\end{array}\right),\qquad\boldsymbol{\mathbf{b}}=\frac{1}{\sqrt{h_2}}\left(\begin{array}{c}
0\\
1\\
\zeta_{,y}
\end{array}\right),\qquad\mathbf{n}=\frac{1}{\sqrt{h}}\left(\begin{array}{c}
-\zeta_{,x}\\
-\zeta_{,y}\\
1
\end{array}\right),\label{sbn}
\end{equation}
where
\[
h_1=\left|\partial_{x}\mathbf{r}_{\Sigma}\right|^{2}=1+\zeta_{,x}^{2},\qquad h_2=\left|\partial_{y}\mathbf{r}_{\Sigma}\right|^{2}=1+\zeta_{,y}^{2},
\]
and 
\[
h=\left|\partial_{x}\mathbf{r}_{\Sigma}\times\partial_{y}\mathbf{r}_{\Sigma}\right|^{2}=1+\zeta_{,x}^{2}+\zeta_{,y}^{2}.
\]
Note that for a 2-D surface, $\boldsymbol{\mathbf{s}}$ and $\boldsymbol{\mathbf{b}}$
are in general not orthogonal as
\begin{equation}
\alpha=\mathbf{s}\cdot\mathbf{b}=\frac{\zeta_{,x}\zeta_{,y}}{\sqrt{h_1 h_2}}\label{asb}
\end{equation}
vanishes only at crests, troughs and saddles. We also consider the intrinsic
curvilinear coordinates $s$ and $b$ on the surface (see Fig.~\ref{FIG1}) defined as
\[
s(x,y)=\int_{0}^{x}\sqrt{h_1(x',y)}\mathrm{d}x',\qquad b(x,y)=\int_{0}^{y}\sqrt{h_2(x,y')}\mathrm{d}y',
\]
and the infinitesimal arc-lengths
\begin{equation}
\mathrm{d}s=\sqrt{h_1}\mathrm{d}x,\quad\quad \mathrm{d}b=\sqrt{h_2}\mathrm{d}y.
\label{dsdb}
\end{equation}

In the global frame, the components $(u_x,u_y)$ of the horizontal particle velocity $\mathbf u_h=(u_x,u_y)$ 
components are denoted by
\begin{equation}
u_{x}=\dot{x},\qquad u_{y}=\dot{y}.
\label{xy}
\end{equation}
 The vertical particle velocity, dictated by the free-surface motion, is given by
\begin{equation}
\dot{\zeta}=\frac{{\rm d}\zeta}{{\rm d}t}=\zeta_{,t}+\dot{x}\zeta_{,x}+\dot{y}\zeta_{,y}=\zeta_{,t}+u_{x}\zeta_{,x}+u_{y}\zeta_{,y}.
\label{zetadot}
\end{equation}
The particle velocity vector written in the global coordinate frame,
\begin{equation}
\mathbf{u}=u_{x}\mathbf{i}+u_{y}\mathbf{j}+\dot{\zeta}\mathbf{k},\label{uu}
\end{equation}
must coincide with its expression in the local frame,
\begin{equation}
\mathbf{u}=u_{s}\mathbf{s}+u_{b}\mathbf{b}+u_{n}\mathbf{n},\label{um}
\end{equation}
where $u_{s}$ and $u_{b}$ are the velocity components tangential to the surface, and satisfy
\begin{equation}
u_{s}=\frac{U_{s}-\alpha U_{b}}{1-\alpha^{2}},\qquad u_{b}=\frac{U_{b}-\alpha U_{s}}{1-\alpha^{2}},\label{usub}
\end{equation}
while
\begin{equation}
u_{n}=\mathbf{u\cdot}\boldsymbol{\mathbf{n}}=\frac{-\zeta_{,x}u_{x}-\zeta_{,y}u_{y}+\dot{\zeta}}{\sqrt{h}}
=\frac{\zeta_{,t}}{\sqrt{h}}\label{un}
\end{equation}
is the particle velocity component orthogonal to the surface. 

Here, $U_{s}$ and $U_{b}$ are the projections of $\mathbf{u}$
onto $\mathbf{s}$ and $\mathbf{b}$ respectively, namely 
\begin{equation}
U_{s}=\mathbf{u\cdot}\boldsymbol{\mathbf{s}}=\frac{u_{x}+\dot{\zeta}\zeta_{,x}}{\sqrt{h_1}}=\frac{(1+\zeta_{,x}^{2})u_{x}+\zeta_{,x}\zeta_{,y}u_{y}+\zeta_{,x}\zeta_{,t}}{\sqrt{1+\zeta_{,x}^2}},\label{ut}
\end{equation}
\begin{equation}
U_{b}=\mathbf{u\cdot}\boldsymbol{\mathbf{b}}=\frac{u_{y}+\dot{\zeta}\zeta_{,y}}{\sqrt{h_2}}=\frac{(1+\zeta_{,y}^{2})u_{y}+\zeta_{,x}\zeta_{,y}u_{x}+\zeta_{,y}\zeta_{,t}}{\sqrt{1+\zeta_{,y}^2}}.\label{ub}
\end{equation}
Note that the denominators in Eq.~\eqref{usub} never vanish as, from Eq.~\eqref{asb}, 
\[
1-\alpha^{2}=\frac{h}{h_1 h_2}=\frac{1+\zeta_{,x}^{2}+\zeta_{,y}^{2}}{\left(1+\zeta_{,x}^{2}\right)\left(1+\zeta_{,y}^{2}\right)}>0.
\]
Clearly, $U_{s}$ and $U_{b}$ coincide with $u_{s}$ and $u_{b}$ on
the surface when $\mathbf{s}$ and $\mathbf{b}$ are orthogonal, i.e. $\alpha=0$.
Note that $u_{n}$ vanishes if the surface is steady or in the comoving
frame of a traveling wave. 

\textcolor{black}{
Drawing on~\cite{Longuet_Higgins_curvature1998}, on the assumption
of a zero-stress free surface any line of inviscid fluid particles parallel
to a principal axis of strain must stretch and be in rotation with
angular velocity $\frac{1}{2}\mathbf{\boldsymbol{\omega}}$, where
$\boldsymbol{\omega}$ is the vorticity vector. Since one axis of
strain is always normal to the free surface, the unit normal
$\mathbf{n}$ 
rotate according to
\begin{equation}
\frac{{\rm d}\mathbf{n}}{{\rm d}t}=\frac{1}{2}\boldsymbol{\mathbf{\omega}}\times\mathbf{n}.\label{dndt}
\end{equation}
}
We then decompose the vorticity as 
\[
\mathbf{\boldsymbol{\omega}}=\mathbf{\mathbf{\boldsymbol{\omega}}}_{\mathrm{\parallel}}+\omega_{\perp}\mathbf{n},
\] 
into its tangential component $\pmb{\omega}_\parallel$ and its normal component
$\omega_\perp\mathbf n$ to the surface.
Note that 
\begin{equation}
\mathbf{n}\times\left(\mathbf{\boldsymbol{\omega}}\times\mathbf{n}\right)=\left(\mathbf{n}\cdot\mathbf{n}\right)\mathbf{\boldsymbol{\omega}}-\left(\mathbf{n}\cdot\boldsymbol{\omega}\right)\mathbf{n}=\mathbf{\boldsymbol{\omega}}-\omega_{\perp}\mathbf{n}=\mathbf{\mathbf{\boldsymbol{\omega}}}_{\mathrm{\parallel}},\label{OMs}
\end{equation}
gives the vorticity aligned along the surface.
The normal vorticity $\omega_{\perp}\mathbf{n}$ cannot be generated
by the surface motion, but it depends upon both the fluid flows above
and below the surface. For example, for irrotational and inviscid
water wave fields $\mathbf{\omega}_{\perp}=0$ as there is no discontinuity
across the surface since vorticity is divergence-free. However,
there is no restriction on the vorticity $\boldsymbol{\omega}_{\parallel}$
generated by the surface motion, which is indeed discontinuous as
it is stored in a vortical sheet along the surface. From Eqs.
\eqref{dndt} and \eqref{OMs} the tangential component $\mathbf{\pmb{\omega}}_{_\parallel}$ of vorticity generated on the free surface is given by (\cite{Longuet_Higgins_curvature1998})
\begin{equation}
\pmb{\omega}_{\parallel}=2\mathbf{n}\times\frac{{\rm d}\mathbf{n}}{{\rm d}t}.\label{wpp}
\end{equation}
From Eq. (\ref{sbn}),
\[
\frac{{\rm d}\mathbf{n}}{{\rm d}t}=\frac{\mathbf{a}}{\sqrt{h}}-\frac{\dot{h}}{2h}\mathbf{n},
\]
where
\[
\mathbf{a}=-\left(\begin{array}{c}
\partial_{x}\dot{\zeta}\\
\partial_{y}\dot{\zeta}\\
0
\end{array}\right),
\]
$\dot{h}=2\nabla\zeta\cdot\nabla\dot{\zeta}$ and $\nabla=\left(\partial_{x},\partial_{y}\right)$
is the 2-D space gradient. Thus, Eq.~\eqref{wpp} yields
\begin{equation}
\pmb{\omega}_{\parallel}=2\mathbf{n}\times\frac{\mathbf{a}}{\sqrt{h}}
=\frac{2}{h}\left(\begin{array}{c}
\partial_{y}\dot{\zeta}\\
\\
-\partial_{x}\dot{\zeta}\\
\\
\zeta_{,x}\partial_{y}\dot{\zeta}-\zeta_{,y}\partial_{x}\dot{\zeta}
\end{array}\right).\label{Ws}
\end{equation}
The $z$-component $\omega_{3}$ can be written in the compact form 
\textcolor{black}{
\begin{equation}
\omega_{3}=\boldsymbol{\omega}_{\parallel}\cdot\mathbf{k}=\frac{2}{h}\left(\zeta_{,x}\partial_{y}\dot{\zeta}-
\zeta_{,y}\partial_{x}\dot{\zeta}\right)=\frac{2}{1+\left|\nabla\zeta\right|^{2}}(\nabla\zeta\times\nabla\dot{\zeta})\cdot {\bf k}\label{W3},
\end{equation}
}
where $\dot{\zeta}$ follows from Eq.~\eqref{zetadot}. This observation is useful to interpret certain terms of the 
symplectic 2-form given in Section~\ref{sec:symplectic2}:  The $b_3$ component of ${\bf b}_{\rm m}$ can be written as
\begin{equation}
b_3=\frac{\nabla \zeta \times \nabla (u_x\zeta_{,x}+u_y\zeta_{,y}+\zeta_{,t})}{1+\vert \nabla\zeta \vert^2}\cdot {\bf k},\label{b3}
\end{equation}
where we have used the two-dimensional cross-product. Comparing Eq.~\eqref{b3} to Eq.~\eqref{W3}, we observe that $b_3=\omega_3/2$ is half the vertical $z$~component of the vorticity created on the free-surface $z=\zeta(x,y,t)$. Note that $b_3$ vanishes when the kinematic criterion~\eqref{b41} for wave breaking holds. We will not dwell too much on the geometric meaning of the components $b_1$ and $b_2$. We only point out that one can show that $b_1$ ($b_2$) is the $z$-component of space-time vorticity created on the space-time surface $z=\zeta(x,y,t)$. Thus, if we imagine trajectories $\overline{\bf z}(\tau)$ as those of ``phase-space parcels'' transported   by the Hamiltonian flow velocity  ${\rm d}\overline{\bf z}/{\rm d}\tau$, then the vector ${\bf b}_{\rm m}$ can be interpreted as space-time vorticity generated by the Hamiltonian flow. These observations will be useful below to interpret the symplectic forms associated with the Hamiltonian equations. 

In the local frame 
\begin{equation}
\mathbf{\boldsymbol{\omega}}_{\parallel}=\omega_{s}\mathbf{s}+\omega_{b}\mathbf{b},\label{omegapar}
\end{equation}
where
\begin{equation}
\omega_{s}=\frac{\Omega_{s}-\alpha\Omega_{b}}{1-\alpha^{2}},\qquad\omega_{b}=
\frac{\Omega_{b}-\alpha\Omega_{s}}{1-\alpha^{2}}.
\label{omegasb}
\end{equation}
The quantities $\Omega_{s}$ and $\Omega_{b}$ are the projections of
$\mathbf{\mathbf{\boldsymbol{\omega}}_{\parallel}}$ onto $\mathbf{s}$
and $\mathbf{b}$, respectively. That is
\begin{equation}
\Omega_{s}=\mathbf{\mathbf{\boldsymbol{\omega}}_{\parallel}\cdot}\boldsymbol{\mathbf{s}}=
2\frac{\sqrt{h_1}\partial_{y}\dot{\zeta}-\alpha\sqrt{h_2}\partial_{x}\dot{\zeta}}{h}
\label{ws},
\end{equation}
and
\begin{equation}
\Omega_{b}=\mathbf{\mathbf{\boldsymbol{\omega}}_{\parallel}\cdot}\boldsymbol{\mathbf{b}}
=-2\frac{\sqrt{h_2}\partial_{x}\dot{\zeta}-\alpha\sqrt{h_1}\partial_{y}\dot{\zeta}}{h}.
\label{wb}
\end{equation}
At the points on the surface where $\mathbf{s}$ and $\mathbf{b}$ are orthogonal ($\alpha=0$), $\Omega_{s}$ and 
$\Omega_{b}$ coincide with $\omega_{s}$ and $\omega_{b}$, respectively. 

Vorticity created on the free-surface $\Sigma$ implies that there
is non-zero circulation of the velocity $\mathbf{u}=(u_{x},u_{y},\dot{\zeta})$
along any closed path $\gamma(\mu,t)=\left(x(\mu,t),y(\mu,t),\zeta(x(\mu,t),y(\mu,t))\right)$
on $\Sigma$, parametrized by $\mu$, and it is conserved by Kelvin's
theorem (see, e.g., \cite{Eyink_notes}). From Eqs.~(\ref{zetadot})
and $\mathrm{d}z=\zeta_{,x}\mathrm{d}x+\zeta_{,y}\mathrm{d}y$, the
circulation around $\gamma$ 
\begin{equation}
\oint_{\gamma(t)}\mathbf{u}\cdot\mathbf{dx}=\oint_{\gamma(t)}u_{x}\mathrm{d}x+u_{y}\mathrm{d}y+\dot{\zeta}\mathrm{d}z,\label{CIRC}
\end{equation}
can be expressed in terms of the projections $U_{s}$ and $U_{b}$
of the  particle velocity $\mathbf{u}$ as (see Eqs.~(\ref{ut}) and (\ref{ub}))
\begin{equation}
\oint_{\gamma(t)}\mathbf{u}\cdot\mathbf{dx}=\oint_{\widetilde{\gamma}(t)}\sqrt{h_1}U_{s}\mathrm{d}x+\sqrt{h_2}U_{b}\mathrm{d}y=\oint_{\widetilde{\widetilde{\gamma}}(t)}U_{s}\mathrm{d}s+U_{b}\mathrm{d}b,\label{CIRC1}
\end{equation}
where we have used Eq.~\eqref{dsdb}, and $\widetilde{\gamma}(t)=\left(x(\mu,t),y(\mu,t)\right)$ and $\widetilde{\widetilde{\gamma}}(t)=\left(s(\mu,t),b(\mu,t)\right)$
are the projected paths of $\gamma$ onto the $x-y$ and $s-b$ planes
respectively. %

Comparing Eqs.~\eqref{px1},~\eqref{py1} with Eqs.~\eqref{ut},~\eqref{ub}, we note that the infinitesimal circulation in Eq.~\eqref{CIRC1} can be written in terms of generalized momenta as
\begin{equation}
U_s \mathrm{d}s + U_b \mathrm{d}b=p_x\mathrm{d}x+p_y\mathrm{d}y,
\label{inc}
\end{equation}
where the arclengths $\mathrm{d}s$ and $\mathrm{d}b$ relate to $\mathrm{d}x$ and $\mathrm{d}y$ via Eq.~\eqref{dsdb}. Thus, the scaled generalized momenta $(p_x/\sqrt{h_1},p_y/\sqrt{h_2})$ are equal to the particle velocity projections $(U_s,U_b)$.

\subsection{Kinematic criterion for wave breaking}\label{sec:kinematic}

\textcolor{black}{In this section, we will show that the classical kinematic criterion for wave breaking~\citep{Perlin2013} follows from the condition of vanishing vorticity at a wave crest.}

First, consider the special case of unidirectional waves propagating along
$x$ and the associated 1-D surface $z=\zeta(x,t).$ In this case, $\mathbf{b}=\mathbf{j}$
is aligned along $y$ and orthogonal to $\mathbf{s}$ (see Fig.~\ref{FIG1}).
Then, from Eq.~(\ref{omegasb}) vorticity created on the surface is aligned along $y$ and it is given by
\begin{equation}
\omega_{b}=\Omega_{b}=-\frac{2}{h_1}\partial_{x}\dot{\zeta}=
\frac{2}{1+\zeta_{,x}^2}\left(\zeta_{,xt}+u_{x}\zeta_{,xx}\right).
\label{wb1}
\end{equation}

This can be written as (\cite{Lundgren}) 
\begin{equation}
\omega_{b}=-2\left(\frac{\mathrm{d}u_{n}}{\mathrm{d}s}+u_{s}K\right),\label{WLH}
\end{equation}
where
\[
K=\frac{\zeta_{,xx}}{h_1^{3/2}}=\frac{\zeta_{,xx}}{\left(1+\zeta_{,x}^2\right)^{3/2}},
\]
 is the surface curvature.  The tangential particle velocity $u_{s}$ follows from Eq.~(\ref{ut})
as
\[
u_{s}=U_{s}=\frac{h_1 u_{x}+\zeta_{,x}\zeta_{,t}}{\sqrt{h_1}},
\]
and the rate of change of the normal particle velocity $u_{n}=\zeta_{,t}/\sqrt{h_1}$
along the intrinsic curvilinear cordinates $s$ on the surface is
given by 
\[
\frac{\mathrm{d}u_{n}}{\mathrm{d}s}=\frac{\mathrm{d}u_{n}}{\mathrm{d}x}\frac{\mathrm{d}x}{\mathrm{d}s}=\frac{\zeta_{,xt}}{h_1}-\frac{\zeta_{,t}\zeta_{,x}\zeta_{,xx}}{h_1^2},
\]
where the infinitesimal arclength $\mathrm{d}s=\sqrt{h_1}\mathrm{d}x$ (see Eq.~\eqref{dsdb}). For steady surfaces $u_{n}=0$ and Eq.~\eqref{WLH} reduces to Longuet-Higgins'
(1988) result
\[
\omega_{b}=-2u_{s}K.
\]

Thus, in a comoving frame where travelling waves are steady, at crests
vorticity is positive or counter-clockwise (\cite{Longuet_Higgins_JFM_bores}).
For unsteady surfaces the normal velocity $u_{n}$ does not vanish
as it balances the underneath horizontal water flow leading to convergence
(growing crests) or divergence (decaying crests). In particular, at
a crest of a wave $\frac{\mathrm{d}u_{n}}{\mathrm{d}s}>0$ since the
wave travels forward as a result of the downward (upward) mass flow
before (after) the crest. Thus, the convergence/divergence of the
flow induced by unsteady surfaces creates negative vorticity that
can counterbalance that generated by the surface curvature. Indeed,
from Eq.~\eqref{wb1} vorticity vanishes at a crest, where $\zeta_{,x}=0$, when
\begin{equation}
\zeta_{,xt}+u_{x}\zeta_{,xx}=0.\label{B1}
\end{equation}

A physical interpretation of this condition is as follows. Consider the horizontal speed ${V}_{c}=\dot{X}_{c}$ of a crest 
located at $X_{c}(t)$ at time $t$. Since at a crest $\zeta_{,x}=0$, we have~\citep{Fedele2014_EPL}
\[
\frac{\mathrm{d}}{\mathrm{d}t}\zeta_{,x}(X_{c}(t),t)=\zeta_{,xt}+\dot{X}_{c}\zeta_{,x}=0,
\]
which implies
\begin{equation}
V_{c}=\dot{X}_{c}=-\frac{\zeta_{,xt}}{\zeta_{,xx}}.\label{Vc1}
\end{equation}
Thus, condition~\eqref{B1}  of vanishing vorticity holds when 
\begin{equation}
u_{x}=V_{c},\label{b2}
\end{equation}
or equivalently when the horizontal particle velocity $u_{x}$ equals the horizontal crest speed $V_{c}$. 

A similar result holds in three dimensions. From Eq.~\eqref{omegapar} vorticity created on a 2-D surface vanishes when
\begin{equation}
\zeta_{,xt}+u_{x}\zeta_{,xx}+u_{y}\zeta_{,xy}=0,\qquad\zeta_{,yt}+u_{x}\zeta_{,xy}+u_{y}\zeta_{,yy}=0,\label{b4}
\end{equation}
or equivalently when the horizontal particle velocity ${\mathbf{u}}_h=(u_{x},u_{y})$ equals the horizontal crest speed $\mathbf{V}_{c}=(\dot{X}_{c},\dot{Y}_{c})$, where $(X_{c}(t),Y_{c}(t))$ is the horizontal crest position. At a crest where $\nabla\zeta=\mathbf{0}$
\[
\frac{\mathrm{d}}{\mathrm{d}t}\nabla\zeta\left(X_{c}(t),Y_{c}(t),t\right)=\nabla\dot{\zeta}=\mathbf{0},
\]
or equivalently
\begin{equation}
\zeta_{,xt}+\dot{X}_{c}\zeta_{,xx}+\dot{Y}_{c}\zeta_{,xy}=0,\qquad\zeta_{,yt}+\dot{X}_{c}\zeta_{,xy}+\dot{Y}_{c}\zeta_{,yy}=0.\label{b41}
\end{equation}
Clearly, Eq. (\ref{b41}) reduces to condition~\eqref{b4} of vanishing vorticity if
\begin{equation}
\mathbf{u}_h=\mathbf{V}_{c}.
\label{br2}
\end{equation} 

Equations~\eqref{b2}~and~\eqref{br2} are the kinematic thresholds
defined as potential breaking criteria for uni- and multidirectional
water waves (see, for example~\cite{Perlin2013}). In particular,
recent experimental results by~\cite{Shemer2014} and~\cite{Shemer2015}
showed that as the largest crest of a focusing wave group grows in
time the crest speed decreases, while water particles at the crest
accelerate. Spilling breakers appear to occur when the horizontal
particle velocity exceeds the crest speed, thus confirming the kinematic
criterion for the inception of wave breaking (see also \cite{Shemer_kinematic2013,Duncan_JFM2001_spilling_1,Duncan_spilling_profile}).


\subsection{Symplecticity and vorticity}\label{sec:symvor}

To gain some intuition on the meaning of the differential one- and two-forms~\eqref{omega1}~and~\eqref{omega2}, we interpret the high-dimensional vector $\overline{\bf z}=(z_{\alpha})$ as the trajectory of a `fluid parcel' that is transported through the extended phase space by the Hamiltonian flow velocity
\begin{equation*}
{\rm v}_H(\tau)=\frac{{\rm d}\overline{\bf z}}{{\rm d}\tau}=\left(\frac{{\rm d} z_{\alpha}}{{\rm d} \tau}\right), 
\end{equation*}
where $z_{\alpha}$ is any of the non-canonical variables $(x,y,u_x,u_y,t,E)$ and the associated velocity
\[
\frac{{\rm d} z_{\alpha}}{{\rm d} \tau}=\{z_{\alpha},\overline{\mathcal{H}}_c\}_*,
\]
follows from the non-canonical Dirac bracket~\eqref{dbracket} (see also Eq.~\eqref{Fbrack}). Then, the symplectic one-form~\eqref{omega1} associated with the Hamiltonian flow can be interpreted as the circulation of the velocity ${\rm v}_H$ along the infinitesimal path~${\rm d}\overline{\bf z}$. 


On the slice $t=\mbox{const}$ of the extended phase space, the non-canonical one-form~\eqref{omega1} simplifies to
\begin{equation*}
\omega^1=\left(\mathbf B{\mathbf u}_h+\pmb{\alpha}\right)\cdot\mathrm d\mathbf
x,
\label{1f}
\end{equation*}
where we have used the identity in Eq.~\eqref{pmon1} and $\mathbf{u}_h=\left(u_x,u_y\right)$ is the horizontal particle velocity. The
one-form $\omega^1$ is invariant along closed material lines. This implies that if $\gamma(t)$ is a closed material line, the quantity
\begin{equation*}
\mathcal C(t)=\oint_{\gamma(t)}\left(\mathbf B{\mathbf u}_h+\pmb{\alpha}\right)\cdot\mathrm d\mathbf
x,
\label{eq:generalKCT}
\end{equation*}
is constant, i.e., it does not vary in time. Clearly, $\mathcal C(t)$ is the physical circulation induced by the particle motion given in Eq.~\eqref{CIRC1}, and is conserved by Kelvin's theorem (see, e.g.~\cite{Eyink_notes}).

Furthermore, on $t=\rm{const.}$ slices, the non-canonical two-form~\eqref{omega2} reduces to
\begin{eqnarray}
\omega^{2} & = & -(1+\vert\nabla\zeta\vert^{2})b_{3}{\rm d}x\wedge{\rm d}y+(1+\zeta_{,x}^{2}){\rm d}u_{x}\wedge{\rm d}x\nonumber\\
 &  & +\zeta_{,x}\zeta_{,y}{\rm d}u_{y}\wedge{\rm d}x+\zeta_{,x}\zeta_{,y}{\rm d}u_{x}\wedge{\rm d}y+(1+\zeta_{,y}^{2}){\rm d}u_{y}\wedge{\rm d}y. \label{omega2a}
\end{eqnarray}
Note that the coefficient $b_3$ of ${\rm d}x\wedge{\rm d}y$ is half the vertical component of the physical vorticity created on the slanted infinitesimal area $dS=(1+\vert\nabla\zeta\vert^{2}){\rm d}x\wedge{\rm d}y$ of the free surface $z=\zeta(x,y,t)$ [see Eqs.~\eqref{b3} and~\eqref{W3}]. 

In Section~\ref{sec:kinematic} we have shown that vorticity vanishes at a surface crest, where $\zeta_{,x}=\zeta_{,y}=0$, when the horizontal particle velocity $\mathbf{u}_{h}$ equals the propagation speed $\mathbf{V}_c$ of the crest [see Eq.~\eqref{b41}], or equivalently when the kinematic criterion~\eqref{b4} for wave breaking holds. In this case the two-form~\eqref{omega2a} further simplifies to
\begin{equation}
\omega^2={\rm d}u_x\wedge {\rm d}x+ {\rm d}u_y\wedge {\rm d}y,
\label{om}
\end{equation}
and the associated Hamiltonian~\eqref{Hc} reduces to
\begin{equation}
\overline{\mathcal H}_c=\frac{u_x^2+u_y^2+\zeta_{,t}^2}{2}+\gr \zeta+E.\label{HH}
\end{equation}
This implies that if the kinematic criterion~\eqref{b4} is attained at the largest crest, i.e.\ when $\zeta_{,t}=0$, then the two-form~\eqref{om} and the associated Hamiltonian $\overline{\mathcal H}_c$ in~\eqref{HH} are those of a particle in free-flight, as if the surface on which the motion is constrained is non-existent and, as a result, vorticity is not created. 
Clearly, in realistic oceanic waves the large crest eventually breaks and energy of fluid particles is dissipated to turbulence as a clear manifestation of time irreversibility. This appears analogous to a flight--crash event in fluid turbulence, where a particle flies with a large velocity before suddenly losing energy~\citep{Falkovich2014}. Clearly, the Hamiltonian particle kinematics associated with the Euler or Zakharov~(1968) equations is time-reversible~\citep{Chabchoub2014} and it may reveal the inviscid mechanism of breaking inception before turbulent dissipative effects take place. To do so, the fluid particle kinematics on the free-surface must be coupled with the dynamics of the irrotational wave field that generates the surface exploiting Zakharov's (1968) Hamiltonian formalism.

\section{Crest slowdown and wave breaking}\label{sec:slow}
\textcolor{black}{In this section, we discuss the relevance of the kinematic criterion for wave breaking~\citep{Perlin2013,Shemer2014,Shemer2015}. Recent studies point at the crest slowdown as what appears to be the underlying inviscid mechanism from which breaking onset initiates. In particular, the multifaceted study by~\cite{Banner_PRL2014} on unsteady highly nonlinear wave packets 
highlights
the existence of a generic oscillatory crest leaning mode that leads
to a systematic crest speed slowdown of approximately $20\%$ lower than the
linear phase speed at the dominant wavelength~(\cite{Fedele2014_EPL}, see also~\cite{Shemer2014}). This explains why initial
breaking wave crest speeds are observed to be approximately $80\%$ of the linear
carrier-wave speed (\cite{RappMelville,Stansell_MacFarlaneJPO2002}).}

Both the particle kinematics on the free surface and the energetics
of the wave field that generates the surface should be considered to establish if the
kinematic criterion for incipient
breaking is valid. 
Recent studies show that the breaking onset of the largest crest of unsteady wave groups initiates before the horizontal particle velocity $u_x$ reaches the crest speed $V_c$, with $x$ being the direction of wave propagation. More specifically, it has been
observed that wave breaking initiates when the particle velocity reaches about $0.84$ times 
the crest velocity (\cite{Barthelemy2015,BannerSaket2015},~see also~\cite{KurniaVanGroesen2014}). 
In fact, none of the recurrent groups reach the threshold $B_x=u_x/V_c=0.84$, while all marginal breaking cases 
exceed the threshold. 

\cite{SongBannerJPO2002}, and more recently~\cite{Barthelemy2015}, explored the 
existence of an energy flux threshold related to the breaking onset. This suggests to look at the space-time transport of 
wave energy fluxes near a large crest of an unsteady wave group and possible local superharmonic instabilities that initiate 
as the threshold $B_x$ is exceeded leading to breaking, as those found for steady steep 
waves~\citep{Longuet-HigginspartI1978}. 

In the following we study the wave energy transport below a crest and the relation to the crest 
slowdown. The irrotational Eulerian velocity field $\mathbf{U}=(U,V,W)=(\phi_{,x},\phi_{,y},\phi_{,z})$
that generates the free surface $\zeta$ is given by the gradient
of the potential $\phi$. From Eq.~\eqref{uu} the velocity $\mathbf u=(u_x,u_y,u_z)$ of a fluid particle that at time $t$ 
passes through the point $\mathbf x_P$  is $\mathbf u(t)=\mathbf U(\mathbf x_P,t)$. 
Besides the Laplace equation to impose
fluid incompressibility in the flow domain, $\phi$ satisfies the
dynamic Bernoulli and kinematic conditions on the free surface (see,
e.g., \cite{Zakharov1968,Zakharov1999}) 
\begin{equation}
\rho\phi_{,t}+\rho \gr\zeta+K_{e}=0,\qquad z=\zeta,\label{B}
\end{equation}
and 
\begin{equation}
\phi_{,z}=\zeta_{,t}+U\zeta_{,x}+V\zeta_{,y},\qquad z=\zeta,\label{B1a}
\end{equation}
where $K_{e}=\rho\mathbf{\left|U\right|}^{2}/2$ is the kinetic energy
density. Drawing on \cite{Tulin2007}, consider the transport equation
\begin{equation}
\partial_{t}K_{e}+\nabla\cdot\mathbf{F}_{K_{e}}=0\label{Ke}
\end{equation}
and the associated flux 
\begin{equation}
\mathbf{F}_{K_{e}}=-\rho\phi_{,t}\mathbf{U}.\label{Flux}
\end{equation}

Equation~(\ref{Ke}) can be written as
\begin{equation}
\partial_{t}K_{e}+\nabla\cdot\left(\mathbf{C}_{K_{e}}K_{e}\right)=0,\label{Ke-1}
\end{equation}
where we have defined the Eulerian kinetic energy flux velocity 
\begin{equation}
\mathbf{C}_{K_{e}}=\frac{\mathbf{F}_{K_{e}}}{K_{e}}=-\frac{\rho\phi_{,t}}{K_{e}}\mathbf{U}.\label{cflux}
\end{equation}
At the free-surface, the kinetic energy flux in Eq.~(\ref{Flux}) can be written as
\begin{equation}
\mathbf{F}_{K_{e}}=\mathbf{U}\left(\rho g\zeta+K_{e}\right),\qquad z=\zeta,\label{Flux-1-1}
\end{equation}
where we have used the Bernoulli equation (\ref{B}). Then, the rate of
change of the surface potential energy density $P_{e}=\rho g\zeta^{2}/2$ \citep{Tulin2007} 
\begin{equation}
\partial_{t}P_{e}=\mathbf{F}_{K_{e}}\cdot\mathbf{n}/\cos\theta\label{Pe}
\end{equation}
is due to the flux of kinetic energy into the moving interface $\zeta$
\begin{equation}
\mathbf{F}_{K_{e}}\cdot\mathbf{n}=U_{n}\left(\rho g\zeta+K_{e}\right)\qquad z=\zeta,\label{Peflux}
\end{equation}
where $U_{n}=\mathbf{U}\cdot\mathbf{n}$ is the fluid velocity normal to the surface and $\theta$
the angle between $\mathbf{n}$ and the vertical (at a wave crest, $\theta=0$). The sum of the total
kinetic energy $K_{e}$ integrated over the wave domain and the potential
energy $P_{e}$ integrated over the surface is conserved. Clearly, a wave crest grows when the adjacent kinetic energy flux behind the crest is larger than the flux after the crest. 

For unidirectional waves $\zeta(x,t)$, the kinematic condition (\ref{B1a}) reduces
to 
\[
\zeta_{,t}=W-U\zeta_{,x},
\]
and
\[
\zeta_{,xt}=\partial_x W+\partial_z W\zeta_{,x}-\partial_z U\zeta_{,x}^{2}-U\zeta_{,xx}.
\]
Then, at $\zeta_{,x}=0$ the crest speed in Eq.~\eqref{Vc1} can be written
as
\begin{equation}
V_{c}=-\frac{\zeta_{,xt}}{\zeta_{,xx}}=U-\frac{\partial_x W}{\zeta_{,xx}}=U-\frac{\partial_z U}{\zeta_{,xx}},\label{VCC}
\end{equation}
where $\partial_x W=\partial_z U$ because of irrotationality. At a crest $\zeta_{,xx}<0$
and the vertical gradient $\partial_z U>0$ as indicated by measurements
and simulations~\citep{Barthelemy_slowdown2015,Barthelemy2015}. As a result, for smooth wave fields the crest speed
$V_{c}$ is always larger than the horizontal fluid velocity $U$.
According to Eq.~(\ref{VCC}), only when crest becomes steep ($\left|\zeta_{,xx}\right|\gg1$)
or the horizontal velocity profile flattens near the crest ($\partial_z U\ll1$)
is the crest speed $V_{c}$ closer to the particle speed $u_x=U$. Thus, the
observation that the initiation of breaking occurs when $V_{c}$ is actually $0.84$ times the particle speed is the kinematic manifestation
of the space-time transport of kinetic energy below the crest~\citep{Barthelemy2015,BannerSaket2015}. Indeed,
from Eq.~(\ref{Peflux}) the normal velocity $C_{K_{e}}$ of kinetic
energy into the moving surface is given by 
\begin{equation}
C_{K_{e}}=\frac{\mathbf{F}_{K_{e}}\cdot\mathbf{n}}{K_{e}}=U_{n}\left(1+\frac{\rho \gr\zeta}{K_{e}}\right).\label{Peflux-1}
\end{equation}
At a crest, where $\zeta>0$, $C_{K_{e}}$ is always larger than the fluid speed $U_{n}$
normal to the surface. However, we expect that as the wave crest grows
reaching nearly breaking the local kinetic energy $K_{e}$ increases
much faster than the potential energy $\rho \gr\zeta$ and $C_{K_{e}}$
tends to $U_{n}$ and the accumulation of potential energy
into the surface is largerly attenuated. Equivalently, the Lagrangian kinetic energy flux
speed $C_{K_{e}}-U_{n}$ seen by fluid particles on the surface is practically null.

\section{There are no finite-time blowups}\label{sec:blowup}
In the appendix of~\cite{sclav05}, contributed by Bridges, the possibility of finite-time blowup
of solutions of the JS equations is discussed.
Bridges studies the special case of the particle kinematics on a 1D surface, i.e., when $\zeta_{,y}\equiv 0$. The equations of motion in Eqs.~\eqref{JS} then reduce to 
\begin{subequations}
\begin{equation*}
\dot x=u_x,\quad \dot y=u_y,
\end{equation*}
\begin{equation*}
\dot u_x=-\frac{2\zeta_{,x}\zeta_{,xt}}{1+\zeta_{,x}^2}u_x-\frac{\zeta_{,x}\zeta_{,xx}}{1+\zeta_{,x}^2}u_x^2
-\frac{\zeta_{,x}(\zeta_{,tt}+\gr)}{1+\zeta_{,x}^2},
\quad
\dot u_y=0.
\end{equation*}
\label{eq:bridges}
\end{subequations}
It is then argued that under the further simplifying assumption that the matrix
\begin{equation*}
\frac{\zeta_{,x}}{1+\zeta_{,x}^2}
\begin{pmatrix}
\zeta_{,xx} & \zeta_{,xt}\\
\zeta_{,xt} & \zeta_{,tt}+\gr\\
\end{pmatrix},
\label{eq:briddges_assume}
\end{equation*}
is constant along trajectories $x(t)$, the horizontal velocity $u_x$ is likely to grow unbounded in finite time.

These assumptions are highly specific and unrealistic.
Nevertheless, Bridges' observation raises the fundamental question
of whether the JS equations are well-posed.
In fact, the right-hand-side of the JS equations (cf. Eq.~\eqref{JS}) is not Lipschitz continuous due to the presence of the quadratic terms
in $\dot{x}$ and $\dot{y}$. Therefore, the elementary results from
ODE theory (i.e., Picard's existence and uniqueness theorem) do not
rule out the finite-time blowup scenario. \textcolor{black}{Note that even though the
free surface $\zeta$ is bounded, the particle velocities obtained from the JS equations~\eqref{JS}
could in principle have a singular behavior.}

Our Hamiltonian formulation for the 3-D particle kinematics
shows that for smooth steady surfaces (i.e., when $\zeta=\zeta(x,y)$
has bounded partial derivatives), the finite-time blowup never occurs.
As we show in Appendix~\ref{app:u2}, the mere conservation of a Hamiltonian
function does not generally rule out the finite-time blowup. 
However, the particular form of the Hamiltonian function~\eqref{HV}
leads to a finite bound on particle speed. 

To see this, note that the Hamiltonian $\mathcal{H}_c=\mathcal{H}_c(\mathbf{x},{\mathbf{p}})$
derived in Eq.~(\ref{HV}) is conserved along the trajectories 
$(\mathbf{x}(t),\mathbf{p}(t))$.
More precisely, 
\begin{equation}
\mathcal{H}_c(\mathbf{x}(t),{\mathbf{p}}(t))=
\mathcal{H}_c(\mathbf{x}(0),{\mathbf{p}}(0))=\mathcal{H}_{0}<\infty,
\label{eq:dHdt=0}
\end{equation}
for all $t$ and finite initial data $(\mathbf x(0),{\mathbf p}(0))$.

On the other hand,
\begin{equation}
\mathcal H_c(\mathbf x,{\mathbf p})=
\frac{1}{2}{\mathbf p}\cdot \mathbf B^{-1}{\mathbf p}+\gr \zeta(\mathbf x)
\geq \frac{|{\mathbf p}|^2}{2(1+|\nabla\zeta(\mathbf x)|^2)}+\gr \zeta(\mathbf x),
\label{eq:Ham_ineq}
\end{equation}
where the inequality follows from the fact that $\mathbf B^{-1}$ is symmetric, positive-definite
with the smallest eigenvalue equal to $(1+|\nabla\zeta|^2)^{-1}$.

Now assume that there exists a finite time $t_{0}$
such that $\lim_{t\to t_{0}}|{\mathbf{p}}(t)|=\infty$, i.e., there
is a blowup at time $t_{0}$. Since $\zeta$ and $|\nabla\zeta|$ are bounded,
inequality~\eqref{eq:Ham_ineq} implies that $\lim_{t\to 
t_{0}}\mathcal{H}_c(\mathbf{x}(t),{\mathbf{p}}(t))=\infty$.
This, however, contradicts the conservation law~\eqref{eq:dHdt=0}.

By definition of the canonical momentum \eqref{pmon1}, 
we have ${\mathbf p}=\mathbf B{\mathbf u}_h$. This in turn implies
$$|{\mathbf p}|^2={\mathbf u}_h\cdot \mathbf B^2{\mathbf u}_h\geq |{\mathbf u}_h|^2,$$
where the inequality follows from the fact that $\mathbf B$ is positive definite
with the smallest eigenvalue equal to $1$. Since $|{\mathbf p}|$ is 
bounded, so is $|{\mathbf u}_h|$, ruling out the finite-time blowup for the particle velocity.
In summary, in the autonomous case (where the smooth surface $\zeta$
is time-independent) the equations of motion~\eqref{JS} are well-posed and finite-time blowup 
cannot occur. 

For traveling waves, 
i.e., $\zeta(x,y,t)=\overline\zeta(x-ct,y)$, one can also show that there are no finite-time blowups. The 
proof 
is similar to the steady case, except that for the traveling waves the conserved Hamiltonian is 
given by Eq.~\eqref{Htw3}.
Namely, in the co-moving frame $\overline{\mathbf x}=(x-ct,y)$, we have
\begin{equation*}
\mathcal H_c(\overline{\mathbf  x}(t),{\mathbf p}(t))\geq 
\frac{|{\mathbf p}-\pmb{\alpha}-c\mathbf B\mathbf e_1|^2}{2(1+|\nabla\zeta(\overline{\mathbf 
x})|^2)}+
\gr\overline\zeta(\overline{\mathbf x})-\frac{1}{2}c^2.
\end{equation*}
As in the steady case, blowup of ${\mathbf p}$ violates the conservation of the Hamiltonian 
function.

For the general non-autonomous case, where $\zeta$ is time-dependent, the finite-time blowup
may not be ruled out by the above argument. 

\section{Trapping regions for steady flows and traveling waves}\label{sec:trapping}
As mentioned earlier, the JS equations are very general as they describe the friction-less motion of a particle on a given surface. Using the Hamiltonian structure in Eq.~\eqref{HV}, we show that 
the horizontal motion of a particle on a steady surface (i.e., $\zeta=\zeta(x,y)$) or on a traveling wave
(i.e., $\zeta=\overline{\zeta}(x-ct,y)$)
is always trapped in a subset of the
two-dimensional  $x-y$ plane. Since the Hamiltonian is conserved, the phase space 
$(x,y,u_x,u_y)\in \mathbb R^4$ 
is foliated by the invariant hypersurfaces $\mathcal{H}=\mbox{const.}$
These hypersurfaces are three-dimensional, and therefore, the particle trajectories can be chaotic. 
It turns out that one can deduce more from the Hamiltonian structure. Namely, we show that, based on their initial
conditions, the trajectories are confined to a subset of the configuration space $(x,y)$. 

We first consider the steady case $\zeta_{,t}=0$, where the Hamiltonian~\eqref{Hc} can be written as
\begin{equation}
\mathcal H(\mathbf x,\mathbf u)=\gr \zeta(\mathbf x)+\frac{1}{2}|\mathbf u|^2+\frac{1}{2}|\mathbf u\cdot \nabla\zeta (\mathbf x)|^2.
\label{eq:simple_Ham}
\end{equation}
Note that the energy $E$ is omitted since the system is autonomous. In this steady case, the following result holds.

\begin{theorem}
Consider the motion of a particle constrained to the smooth steady surface 
$\zeta=\zeta(\mathbf x)$. 
Denote the initial condition of the particle by $(\mathbf x_0,\mathbf u_0)$ and define 
\begin{equation}
D_0:=\left\{\mathbf x=(x,y)\in\mathbb R^2| \zeta(\mathbf x)\leq \zeta(\mathbf x_0)+
\frac{1}{2\gr}|\mathbf u_0|^2+\frac{1}{2\gr}|\mathbf u_0\cdot \nabla\zeta (\mathbf x_0)|^2\right\}.
\label{D0}
\end{equation}
The position of the particle is bound to the subset $D_0$, i.e., $(x(t),y(t))\in D_0$ for all
times $t$.
\label{thm:trapping}
\end{theorem}
\begin{proof}
Hamiltonian~\eqref{eq:simple_Ham} is conserved along particle trajectories $(\mathbf x(t),\mathbf u(t))$. 
Hence we have
$$\gr \zeta(\mathbf x(t))\leq \mathcal H(\mathbf x(t),\mathbf u(t))=\mathcal H(\mathbf x_0,\mathbf u_0).$$
\end{proof}

Note that the above theorem does not imply that the subset $D_0$ is invariant. In fact, particles 
initiated outside $D_0$ 
can very well enter (and exit) the set. Instead, the set $D_0$ is a \emph{trapping region}, i.e., 
particles starting in $D_0$ with initial conditions $(\mathbf x_0,\mathbf u_0)$ stay
in $D_0$ for all times. For a given surface, the trapping region $D_0$ is entirely determined by 
the initial position $\mathbf x_0$ and the initial velocity $\mathbf u_0$ of the particle.

An interesting special case is to consider the motion of the particle from rest, i.e., zero initial velocity.
Then Theorem~\ref{thm:trapping} implies the following.

\begin{corollary}
Consider the motion of a particle that is initially at rest and moves on a smooth steady surface 
$\zeta=\zeta(x,y)$.
Denote the initial position of the particle by $(x_0,y_0)$ and define 
\begin{equation*}
D_0:=\left\{(x,y)\in\mathbb R^2| \zeta(x,y)\leq \zeta(x_0,y_0) \right\}.
\label{D0_0}
\end{equation*}
The position of the particle is bound to the subset $D_0$, i.e. $(x(t),y(t))\in D_0$ for all
times $t$.
\label{cor:trapping}
\end{corollary}

\begin{proof}
This is a direct consequence of Theorem~\ref{thm:trapping} with the initial velocity ${\mathbf u}_0=\mathbf 0$.
\end{proof}

Theorem~\ref{thm:trapping} and Corollary~\ref{cor:trapping} hold for traveling waves, $\zeta(x,y,t)=\overline{\zeta}(x-ct,y)$.
The statements are identical except that the coordinate $x$ and the velocity $u_x$ are replaced with the co-moving
coordinate $\bar x = x-ct$ and velocity ${\bar u}_x=\dot x-c$, respectively. The proofs are similar and therefore omitted here. The trapping region in Eq.~\eqref{D0} is now given by
\begin{equation*}
D_0:=\left\{\mathbf x=(\bar x,y)\in\mathbb R^2| \overline{\zeta}(\mathbf x)\leq \overline{\zeta}(\mathbf x_0)+
\frac{1}{2\gr}|\mathbf u_0-c\mathbf e_1|^2+\frac{1}{2\gr}|(\mathbf u_0-c\mathbf e_1)\cdot \nabla \overline{\zeta} (\mathbf x_0)|^2 \right\},
\label{D01}
\end{equation*}
where $\mathbf e_1$ is the unit vector along $\bar x$ and the initial particle velocity $\mathbf u_0$ is that in the fixed reference frame.
%
\section{Concluding remarks}
We have investigated the properties of the JS equations for the kinematics of fluid particles on the sea surface. We showed that the JS equations can be derived from an action principle describing the motion of a frictionless particle constrained on an unsteady surface and subject to gravity. 
Further, for a zero-stress free surface the classical kinematic criterion for wave breaking is deduced from the condition of vanishing of vorticity generated at a crest. If this holds for the largest crest, the Hamiltonian structure of the JS equations reveals that the associated symplectic two-form instantaneously reduces to that of the motion of a particle in free flight, as if the constraint to be on the free surface did not exist. 

In realistic oceanic fields the large crest eventually breaks and energy of fluid particles is dissipated to turbulence, which is a time-irreversible mechanism. We speculate that this behavior appears analogous to a flight--crash event in fluid turbulence, where a particle flies with a large velocity before suddenly losing energy~\citep{Falkovich2014}. Clearly, the Hamiltonian particle dynamics associated with the inviscid Euler or Zakharov~(1968) equations is time-reversible~\citep{Chabchoub2014}. Then, the instantenous vanishing of vorticity at large crests may reveal the inviscid 
mechanism of breaking inception before turbulent dissipative effects take place. This necessitates a further study of the dynamics and energetics of the wave field that generates the free surface to verify if the kinematic breaking criterion is valid.

Finally, the conservation and special form of the Hamiltonian function for steady surfaces and traveling waves implies that particle velocities remain bounded at all times, ruling out the finite-time blowup of solutions.

\section*{Acknowledgments}
FF acknowledges the Georgia Tech graduate courses `Classical Mechanics II' taught by Jean Bellissard in Spring 2013 and `Nonlinear dynamics: Chaos, and what to do about it?' taught by Predrag Cvitanovi\'c in Spring 2012. FF also thanks Jean Bellissard for stimulating discussions on differential geometry and classical mechanics as well as for a revision of an early draft of the manuscript. The authors are also grateful to Jean Bellissard, Predrag Cvitanovi\'c and Rafael De La Llave  for stimulating discussions on symplectic geometry and Hamiltonian dynamics.

\begin{appendices}

\section{JS equations for steady irrotational flows}\label{app:Stokes}
Consider a one-dimensional, semi-infinite, steady, irrotational flow constrained to the wave surface $\zeta=\zeta(x)$. 
These assumptions imply $\phi_{,y}=\phi_{,t}=0$ (where $\phi$ is the velocity potential) and $\zeta_{,y}=\zeta_{,t}=0$. Since the vertical particle velocity satisfies
$$\dot z=\phi_{,z}(x(t),z(t)),$$
the respective acceleration is given by
$$\ddot z = \phi_{,xz}\dot x + \phi_{,zz}\dot z.$$

For particles on the surface, $z(t)=\zeta(x(t))$, which implies 
$$\dot z = \zeta_{,x} \dot x,$$
and
$$\ddot z = \zeta _{,xx} \dot x^2+\zeta_{,x} \ddot x.$$
Therefore,
$$\zeta _{,xx} \dot x^2+\zeta_{,x} \ddot x=\phi_{,xz}\dot x + \phi_{,zz}\zeta_{,x}\dot x,$$
which upon multiplying by $\zeta_{,x}$ and rearranging terms gives
\begin{equation}
\zeta_{,x}^2 \ddot x=-\zeta _{,xx}\zeta_{,x} \dot x^2+\phi_{,xz}\zeta_{,x}\dot x + \phi_{,zz}\zeta_{,x}^2\dot x.\label{EQ1}
\end{equation}

On the other hand, the Bernoulli equation~\eqref{B} reads
$$\gr\zeta(x(t))+\frac{1}{2}(\dot x^{2}+\phi_{,z}^2(x(t),\zeta(x(t)))=0.$$
Taking the derivative with respect to time we obtain
$$\ddot x=-\gr\zeta_{,x}-\phi_{,z}\big(\phi_{,xz}+\phi_{,zz}\zeta_{,x}\big).$$
Using $\phi_{,z}=\dot z=\zeta_{,x}\dot x$ implies
\begin{equation}
\ddot x=-\gr\zeta_{,x}-\phi_{,xz}\zeta_{,x}\dot x-\phi_{,zz}\zeta_{,x}^2\dot x.
\label{EQ2}
\end{equation}
Adding Eqs.~\eqref{EQ1}~and~\eqref{EQ2} gives the JS equations~\eqref{JS} in the case of 1-D steady flows. 


\section{Computation of the Dirac bracket} \label{app:Dirac}

Since the surface is time-dependent, the resulting constraints have an explicit time-dependence. We first autonomize the system of the free particle in three dimensions adding a pair of canonically conjugate variables $(t,E)$, where $E$ is the energy exchanged by the particle with the moving surface. Indeed, the particle behaves as an open system if the motion is on unsteady surfaces. Constraints are now functions of the dynamical variables
\[
\overline{\bf z}=(x,y,z,t,u_x,u_y,u_z,E),
\]
as required by Dirac's theory, and $\overline{\bf z}(\tau)$ is a generic trajectory in the extended phase space, parametrized by $\tau$ which plays the role of time for the autonomous system.
The autonomized Hamiltonian of the free particle in three dimensions subjected to gravity is
\begin{equation*}
\overline{\mathcal{H}}=\frac{u_x^2+u_y^2+u_z^2}{2}+\gr z +E.\label{Hfree} 
\end{equation*}
The two constraints are given by
\begin{equation*}
\Phi_1=z-\zeta(x,y,t)=0,\quad\quad \Phi_2=u_z-u_x\zeta_{,x}-u_y \zeta_{,y}-\zeta_{,t}=0.
\end{equation*}
The $2\times 2$ matrix ${\mathbb D}$ follows from the inverse of the symplectic matrix ${\mathbb C}$ given by Eq.~(\ref{eq4C}):
$$
D_{11}=D_{22}=0,\quad\quad D_{21}=-D_{12}=1/(1+\zeta_{,x}^2+\zeta_{,y}^2). 
$$
The Poisson matrix associated with the Dirac bracket is computed from 
\begin{equation}
\label{eq:Jstar}
{\mathbb J}_*={\mathbb J}-{\mathbb J}\hat{\cal Q}^\dagger{\mathbb D}\hat{\cal Q}{\mathbb J},
\end{equation}
where the $K\times N$ matrix $\hat{\cal Q}$ has elements
$$
\hat{\cal Q}_{\alpha l}=\frac{\partial \Phi_\alpha}{\partial z_l},
$$ 
and $\dagger$ denotes Hermitian transposition. Since ${\mathbb C}=\hat{\cal Q}{\mathbb J}\hat{\cal Q}^\dagger$, the Poisson matrix of the Dirac bracket can be computed algebraically by way of a projector~\citep{chan13a}
$$
{\cal P}_*={\mathbb I}_N-\hat{\cal Q}^\dagger{\mathbb D}\hat{\cal Q}{\mathbb J},
$$
where ${\mathbb I}_N$ is the $N\times N$ identity matrix. If ${\mathbb C}$ is invertible, ${\cal P}_*\hat{\cal Q}^\dagger=0$, which is an alternative way to characterize the fact that constraints that are actually Casimir invariants of the Dirac bracket. Actually, the matrix ${\mathbb D}$ is defined such that the constraints are Casimir invariants of the Dirac bracket, i.e., $\{F,\Phi_{\alpha}\}_*=0$ for all observables $F$. As a result, the Dirac bracket is a Poisson bracket that satisfies the Jacobi identity~\citep{chan13b}.
The Dirac projector ${\cal P}_*$ projects the dynamics onto the surface defined by the constraints. The expression of ${\mathbb J}_*$ is given by ${\mathbb J}_*={\mathbb J}{\cal P}_*={\cal P}_*^\dagger {\mathbb J}{\cal P}_*$. This provides a systematic and algebraic procedure to compute Dirac brackets. 

The resulting Poisson matrix does not explicitly depend on $z$ and $u_z$. As a consequence, the Poisson bracket of two functions of $(x,y,t,u_x,u_y,E)$ is again a function of $(x,y,t,u_x,u_y,E)$. In other words, the algebra of observables $F(x,y,t,u_x,u_y,E)$ is a Poisson sub-algebra. In this way, one can omit $z$ and $u_z$ (since their dynamics is quite trivially given by the constraints which are Casimir invariants of the Dirac bracket) and the phase-space dimension is reduced by two. This leads to the expression of the Dirac bracket given by Eq.~(\ref{eq:DBzeta}).

\section{The Hamiltonian structure of a prototype blowup problem}\label{app:u2}
As a toy problem, Bridges considers the simplest second order ODE Riccatti equation, which can be written as
\begin{equation}
\dot{x}=u,\quad\quad\dot{u}=u^2.
\label{eq:ftbu}
\end{equation}
Although the JS equations cannot be reduced to this form, we discuss its properties for completeness. The system~\eqref{eq:ftbu} possesses the Hamiltonian
\[
H=u{\rm e}^{-x},
\]
which is of course an invariant of the dynamics. The non-canonical Poisson bracket is given by
$$
\{F,G\}=u {\rm e}^x \left( \frac{\partial F}{\partial x}\frac{\partial G}{\partial 
u}-\frac{\partial F}{\partial u}\frac{\partial G}{\partial x} \right).
$$  
The canonical structure of the system is obtained in the variables $(x,{\rm e}^x \ln u)$. 

For initial conditions $\left(x_{0},u_{0}\right)$ at $t=0$
\[
x(t)=x_{0}+\ln\frac{1}{1-u_{0}t},\qquad u(t)=\frac{u_{0}}{1\text{\textminus}u_{0}t}.
\]
Clearly, for positive initial velocities ($u_{0}>0$), all solutions
blow up in finite time, with the time of blowup inversely proportional
to the norm of the initial velocity data. On the other hand, trajectories
are bounded for negative initial velocities and they exist for all
time. 

The finite-time singularity of the system can be explained exploiting the time invariance of the Hamiltonian $H=u{\rm e}^{-x}$. As $u$ linearly tends to infinity when $t$ tends to some $t_0$, $x$ 
also tends to infinity, but logarithmically, when $t$ goes to $t_0$, in such a way that the product 
between $u$ and ${\rm e}^{-x}$. This is possible because  ${\rm e}^{-x}$ is not bounded from below 
by a strictly positive quantity. Contrast this with Eq.~\eqref{eq:Ham_ineq}, where the quadratic part of the Hamiltonian is positive definite, and hence bounded from below by a positive constant.

%

\end{appendices}

\bibliographystyle{jfm}
\bibliography{biblioFranco}

\end{document}